\documentclass{LMCS}

\def\dOi{11(3:8)2015}
\lmcsheading%
{\dOi}
{1--29}
{}
{}
{Nov.~12, 2014}
{Sep.~\phantom04, 2015}
{}

\ACMCCS{[{\bf Mathematics of computing}]: Probability and
  statistics---Stochastic processes---Markov processes;  Probability
  and statistics---Distribution functions} 

\usepackage{hyperref}
\usepackage{amsmath}
\usepackage{amssymb}
\usepackage{amsfonts}
\usepackage{mathrsfs}
\usepackage{dsfont}

\usepackage{algorithm}
\usepackage{algorithmic}

\usepackage{longtable}

\usepackage{epsfig}
\usepackage[usenames,dvipsnames]{pstricks}

\usepackage{xspace}
\def\eg{{\em e.g.}\xspace}
\def\cf{{\em cf.}\xspace}

\begin{document}

\title[Formal Abstractions of Continuous-Space Markov Processes]{
Quantitative Approximation of the\\ Probability Distribution
of a Markov Process\\ by Formal Abstractions\rsuper*
}

\author[S.~Esmaeil Zadeh Soudjani]{Sadegh Esmaeil Zadeh Soudjani}
\address{Department of Computer Science, University of Oxford, United Kingdom}
\email{\{Sadegh.Soudjani,Alessandro.Abate\}@cs.ox.ac.uk}

\author[A.~Abate]{Alessandro Abate}
\address{\vspace{-18 pt}}
\titlecomment{{\lsuper*}This article generalises and completes the results presented
  in \cite{SA14} and specifically benefits from extensions first
  discussed in \cite{SAH12}.}
\thanks{This work has been supported by the European Commission via STREP project MoVeS 257005 and IAPP project AMBI 324432, 
and by the John Fell Oxford University Press Research Fund.}

\keywords{Continuous-Space Markov Processes, 
Discrete-time Stochastic Systems, 
Formal Abstractions, 
PCTL Verification, 
Probabilistic Invariance, 
Higher-Order Approximations}


\begin{abstract}
\label{sec:abstract}
\noindent The goal of this work is to formally abstract a Markov process evolving in discrete time over a general state space as a finite-state Markov chain, 
with the objective of precisely approximating its state probability distribution 
in time, 
which allows for its approximate, faster computation by that of the Markov chain. 
The approach is based on formal abstractions and employs an arbitrary finite partition of the state space of the Markov process, 
and the computation of average transition probabilities between partition sets.
The abstraction technique is formal, 
in that it comes with guarantees on the introduced approximation that depend on the diameters of the partitions:  
as such, they can be tuned at will.  
Further in the case of Markov processes with unbounded state spaces, 
a procedure for precisely truncating the state space within a compact set is provided, 
together with an error bound that depends on the asymptotic properties of the transition kernel of the original process. 
The overall abstraction algorithm, 
which practically hinges on piecewise constant approximations of the density functions of the Markov process, 
is extended to higher-order function approximations:  
these can lead to improved error bounds and associated lower computational requirements.     
The approach is practically tested to compute probabilistic invariance of the Markov process under study, 
and is compared to a known alternative approach from the literature. 
\end{abstract}

\maketitle

\section{Introduction}
\label{sec:intro}
Verification techniques and tools for deterministic, discrete time, finite-state systems have been available for many years \cite{Kur94}. 
Formal methods in the stochastic context are typically limited to finite-state structures, 
either in continuous or in discrete time \cite{BKH99,KNSS00}.
Stochastic processes evolving over continuous (uncountable) spaces are often related to undecidable problems (the exception being when they admit analytical solutions).
It is thus of interest to resort to formal approximation techniques that allow solving decidably corresponding problems over finite discretisations of the original models.
In order to formally relate the computable approximate solutions to the original problems,  
it is of interest to come up with explicit bounds on the error introduced by the approximations. 
The use of formal approximations techniques over complex models can be looked at from the perspective of research on abstractions, 
which are of wide use in formal verification. 

Successful numerical schemes based on Markov chain approximations of 
general stochastic systems in continuous time have been introduced in the literature \cite{KD01}.
However, the finite abstractions are only related to the original models asymptotically (at the limit, that is weakly), with no explicit error bounds. 
This approach has been applied to the approximate study of probabilistic reachability or safety of stochastic hybrid models in \cite{KR06,PH06}.
An alternative line of work on approximations of continuous-space, discrete-time Markov processes is pursued in \cite{DDP04,DGJP03}, 
where the discussed approximation scheme generates a finite-state model.
In \cite{DDP03} the idea of approximating by averaging is introduced,
where the conditional expectation is used to compute the approximation,
and is later extended in \cite{CDPP14}.
The weak point of these contributions is the fact that the approximation errors that are essential in assessing the quality of the approximation are not computed.
 
As an alternative to qualitative approximations, 
in \cite{APKL10} a technique has been introduced to provide formal abstractions of discrete-time, continuous-space Markov models \cite{APLS08}, 
with the objective of investigating their probabilistic invariance (safety) via model checking procedures over a finite Markov chain.
In view of computational scalability, the approach has been improved and optimised in \cite{SA11,SA13}, 
extended to a wider class of processes \cite{SA12,SATAC12}, and practically implemented as a software tool \cite{FAUST15}.  
These abstraction techniques hinge on piecewise-constant approximations of the kernels of the Markov process. 
Linear projection operators are employed in \cite{SAH12} to generalise these techniques via higher-order interpolations that provide improved error bounds on the approximation level.

In this work we show that the approach in \cite{APKL10,SA13} can be successfully employed to approximately compute the statistics in time of a stochastic process over a continuous state space. 
We first provide a forward recursion for the approximate computation of the state distribution in time of the Markov process. 
This computation is based on a partitioning of the state space, 
and on the abstraction of the Markov process as a finite-state Markov chain. 
Further, a higher-order approximation method is presented, as a generalisation of the approach above, 
and an upper bound on the error related to the new approximation is formally derived.
Based on the information gained from the state distribution in time, 
we show that the method can be used as an alternative to \cite{APKL10,SAH12,SA12,SA13} to approximately compute 
probabilistic invariance (safety) for discrete-time stochastic systems over general state spaces.
Probabilistic invariance (safety) is the dual problem to probabilistic reachability. 
Over deterministic models reachability and safety have been vastly studied in the literature, 
and computational algorithms and tools have been developed based on both forward and backward recursions.
Similarly, for the probabilistic models under study, 
we compare the presented approach (based on forward computations) with existing approaches in the literature \cite{APKL10,SAH12,SA12,SA13} (which hinge on backward computations), 
particularly in terms of the introduced approximation error.

The Markov chain abstraction applied to the forward/backward computation of probabilistic invariance can be generalised to other specifications expressed as non-nested PCTL formulae or to reward-based criteria characterised via value function recursions. 
Moreover, the constructed Markov chain can be shown to represent an \emph{approximate probabilistic bisimulation} of the original process \cite{A12,TA13,AKNP14}.

The article is structured as follows. 
Section~\ref{sec:prelim} introduces the model under study and discusses some structural assumptions needed for the abstraction procedure. 
The procedure comprises two separate parts:
Section~\ref{sec:trunc} describes the truncation of the dynamics of the model, 
whereas Section~\ref{sec:partition} details the abstraction of the dynamics (approximation of the transition kernel) -- 
both parts contribute to the associated approximation error. 
Section \ref{sec:approx&error} considers higher-order approximation schemes and quantifies
the introduced approximation error. Section \ref{sec:proj} specialises 
these higher-order schemes to explicit algorithms for low-dimensional models using known interpolation bases. 
Section~\ref{sec:safety} discusses the application of the procedure to the computation of probabilistic invariance, 
and compares it against an alternative approach in the literature.

\smallskip

In this article we use $\mathbb N \doteq \{1,2,3,\ldots\}$ to denote the natural numbers,
$\mathbb N_m\doteq \{1,2,\ldots,m\}$ and $\mathbb Z_m\doteq \{0,1,2,\ldots,m\}$
for any $m\in\mathbb N$.

\section{Models, Preliminaries, and Goals of this work}
\label{sec:prelim}

We consider a discrete time Markov process $\mathscr M_{\mathfrak s}$ defined over a general state space, 
which is characterised by a pair $(\mathcal S, T_{\mathfrak s})$,
where $\mathcal S$ is a continuous state space that we assume endowed with a metric and be separable\footnote{A metric space $\mathcal S$ is called separable if it has a countable dense subset.}.
We denote by $(\mathcal S,\mathcal B(\mathcal S),\mathcal P)$ the probability structure on $\mathcal S$,
with $\mathcal B(\mathcal S)$ being the Borel $\sigma$-algebra\footnote{The Borel $\sigma$-algebra $\mathcal B(\mathcal S)$ is the smallest $\sigma$-algebra in $\mathcal S$ that contains all open subsets of $\mathcal S$. For a separable metric space $\mathcal S$, $\mathcal B(\mathcal S)$ equals the $\sigma$-algebra generated by the open (or closed) balls of $\mathcal S$.} 
in $\mathcal S$ and $\mathcal P$ a probability measure to be characterised shortly.  
$T_{\mathfrak s}$ is a stochastic kernel that assigns to each point $s\in\mathcal S$ a probability measure $T_{\mathfrak s}(\cdot|s)$, 
so that for any measurable set $A\in\mathcal B(\mathcal S)$, 
$\mathcal P(s(1) \in A|s(0) = s) = T_{\mathfrak s}(A|s)$. 
We assume that the stochastic kernel $T_{\mathfrak s}$ admits a density function $t_{\mathfrak s}$,
namely $T_{\mathfrak s}(A|s) = \int_A t_{\mathfrak s}(\bar s|s)d\bar s$.  

Given the measurable space $\left(\mathcal S,\mathcal B(\mathcal S), \mathcal P\right)$, we set up the product space $\mathcal S^{t+1}$
containing elements $\mathbf s(t) = \left[s(0),s(1),\ldots,s(t)\right]$, where the bold typeset is used in the sequel to indicate vector quantities.
Suppose that the initial state of the Markov process $\mathscr M_{\mathfrak s}$ is distributed according to the density function $\pi_0:\mathcal S\rightarrow \mathbb R^{\ge 0}$. Then the multi-variate density function $\pi_0(s_0)t_{\mathfrak s}(s_1|s_0)t_{\mathfrak s}(s_2|s_1)\ldots t_{\mathfrak s}(s_t|s_{t-1})$
is a probability measure $\mathbb P$ on the product space $\mathcal S^{t+1}$.
On the other hand the state distribution of $\mathscr M_{\mathfrak s}$ at time $t \in \mathbb N$ is characterised by a density function $\pi_t:\mathcal S\rightarrow \mathbb R^{\ge 0}$,
which fully describes the statistics of the process at time $t$ and is in particular such that, 
for all $A \in \mathcal B(\mathcal S)$, 
\begin{equation*}
\mathbb P(s(t)\in A) = \int_{A} \pi_t(s)ds,
\end{equation*}
where the symbol $\mathbb P$ is used to indicate the probability associated to events over the product space $\mathcal S^{t+1}$
(note that the event $s(t)\in A$ is equivalent to $\mathbf s(t)\in\mathcal S^t\times A$ on their corresponding probability spaces).

The state density functions $\pi_{t}(\cdot)$ can be characterised recursively, as follows: 
\begin{equation}
\label{eq:recurs}
\pi_{t+1}(\bar s) = \int_{\mathcal S}t_{\mathfrak s}(\bar s|s)\pi_t(s)ds \quad \forall\bar s \in \mathcal S.
\end{equation}
In practice the forward recursion in \eqref{eq:recurs} rarely yields a closed form for $\pi_t(\cdot)$.  
A special instance where this is the case is represented by a linear dynamical system perturbed by a Gaussian process noise: 
due to the closure property of Gaussian distributions over addition and multiplication by a constant,  
it is possible to explicitly write recursive formulae for the mean and the variance of the distribution, 
and thus express in a closed form the distribution in time of the solution process.  
In more general cases, it is necessary to numerically (hence, approximately) compute this density function in time.    

\medskip

This article provides a numerical approximation of the density function of $\mathscr M_{\mathfrak s}$ in time 
as the probability mass function (pmf) of a finite-state Markov chain $\mathscr M_{\mathfrak p}$.  
The Markov chain $\mathscr M_{\mathfrak p}$ is obtained as an abstraction of the concrete Markov process $\mathscr M_{\mathfrak s}$. 
The abstraction is associated with a guaranteed and tunable error bound,
and algorithmically it leverages a state-space partitioning procedure.  
The procedure is comprised of two steps: 
\begin{enumerate} 
\item 
since the state space $\mathcal S$ is generally unbounded, 
it is first properly truncated; 
\item 
subsequently,  
a partition of the truncated dynamics is introduced. 
\end{enumerate} 

Section~\ref{sec:trunc} discusses the error generated by the state-space truncation, 
whereas Section~\ref{sec:partition} describes the construction of the Markov chain by state-space partitioning.
The discussed Markov chain abstraction is based on a piecewise-constant approximation of the density functions. 
In order to improve the efficiency and the precision of the approximation, 
we generalise the abstraction method in Section~\ref{sec:approx&error} utilizing higher-order approximations of the density functions.
We employ the following example throughout the article as a running case study. 
\begin{exa}
\label{ex:linear_1d}
Consider the one-dimensional stochastic dynamical system
\begin{equation*}
s(t+1) = a s(t) + b +\sigma w(t),
\end{equation*}
where the parameters $a,\sigma>0$, whereas $b \in \mathbb R$, 
and $w(\cdot)$ is a process comprised of independent, identically distributed random variables with a standard normal distribution. 
The initial state of the process is selected uniformly within the bounded interval $[\beta_0,\gamma_0] \subset \mathbb R$. 
The solution of the model is a Markov process, 
evolving over the state space $\mathcal S = \mathbb R$, 
and fully characterised by the conditional density function 
\begin{equation*}
t_{\mathfrak s}(\bar s|s) = \phi_\sigma(\bar s-as-b),\quad \text{ where } \quad \phi_\sigma(u) = \frac{1}{\sigma\sqrt{2\pi}}e^{-u^2/2\sigma^2}.\eqno{\qEd}
\end{equation*}
\end{exa}

We raise the following assumptions in order to
relate the state density function of $\mathscr M_{\mathfrak s}$ to the probability mass function of $\mathscr M_{\mathfrak p}$.
More precisely, 
these assumptions are employed for the computation of the approximation error, 
but the abstraction approach proposed in this article can be applied without raising them.
\begin{asm}
\label{ass:marg_conv}
For given sets $\varGamma\subset\mathcal S^2$ and $\Lambda_0\subset\mathcal S$,
there exist positive constants $\epsilon$ and $\varepsilon_0$, such that 
$t_{\mathfrak s}(\bar s|s)$ and $\pi_0(s)$ satisfy the following conditions: 
\begin{equation}
\label{eq:ineq_dens}
t_{\mathfrak s}(\bar s|s)\le \epsilon, \,\,
\forall (s,\bar s)\in\mathcal S^2\backslash\varGamma, 
\qquad \text{and}\qquad
\pi_0(s)\le \varepsilon_0, \,\,
\forall s\in\mathcal S\backslash\Lambda_0.
\end{equation}
\end{asm}
\begin{asm}
\label{ass:Lip_cont_bounded}
The density functions $\pi_0(s)$ and $t_{\mathfrak s}(\bar s|s)$ are (globally) Lipschitz continuous,  
namely there exist finite constants $\lambda_0, \lambda_{\mathfrak f}$, 
such that the following Lipschitz continuity conditions hold: 
\begin{align}
& |\pi_0(s)-\pi_0(s')|\le\lambda_0 \|s-s'\|, \quad
\forall s,s'\in\Lambda_0,\label{eq:cond_init_dens}\\
& |t_{\mathfrak s}(\bar s|s)-t_{\mathfrak s}(\bar s'|s)|\le \lambda_{\mathfrak f}\|\bar s-\bar s'\|, \quad
\forall s,\bar s,\bar s'\in\mathcal S.
\label{eq:cond_kern}
\end{align}
Moreover, there exists a finite constant $M_{\mathfrak f}$ such that
\begin{equation}
M_{\mathfrak f} = \sup\left\{\int_{\mathcal S}t_{\mathfrak s}(\bar s|s)ds\bigg| \bar s\in\mathcal S\right\}.
\end{equation}
\end{asm}
The Lipschitz constants $\lambda_0,\lambda_{\mathfrak f}$ are practically computed by taking partial derivatives of the density functions $\pi_0(\cdot),t_{\mathfrak s}(\cdot|s)$ and maximising their norm.
The sets $\Lambda_0$ and $\varGamma$ will be used to truncate the support of the density functions 
$\pi_0(\cdot)$ and $t_{\mathfrak s}(\cdot| s)$, respectively. 
Assumption~\ref{ass:marg_conv} enables the precise study of the behaviour of density functions $\pi_t(\cdot)$
over the truncated state space. 
Furthermore, 
the Lipschitz continuity conditions in Assumption~\ref{ass:Lip_cont_bounded} are essential to derive error bounds related to the abstraction of the Markov process over the truncated state space. 
In order to compute these error bounds, 
we assign the infinity norm to the space of bounded measurable functions over the state space $\mathcal S$, namely 
\begin{equation*}
\|f(s)\|_{\infty} =\sup_{s\in\mathcal S}|f(s)|, \quad \forall f\in
\mathbb B(\mathcal S)\doteq\{g:\mathcal S\rightarrow\mathbb R,\text{ }g\text{ bounded and measurable}\}.
\end{equation*} 
In the sequel the function $\mathds 1_{A}(\cdot)$ denotes the indicator function of a set $A \subseteq \mathcal S$, 
namely $\mathds 1_{A}(s) = 1$, if $s \in A$; else $\mathds 1_{A}(s) = 0$. 

\medskip

\textbf{Example \ref{ex:linear_1d} (Continued).}
Select the interval $\Lambda_0 = [\beta_0,\gamma_0]$ and define the set $\varGamma$ by the linear inequality
\begin{equation*}
\varGamma = \{(s,\bar s)\in\mathbb R^2\big| |\bar s-as-b|\le\alpha\sigma\}.
\end{equation*}
The initial density function $\pi_0$ of the process can be represented by the function 
\begin{equation*}
\pi_0(s)=\mathds 1_{[\beta_0,\gamma_0]}(s)/(\gamma_0-\beta_0).
\end{equation*}
Then Assumption~\ref{ass:marg_conv} is valid with constants $\epsilon = \phi_1(\alpha)/\sigma$ and $\varepsilon_0=0$.
The constant $M_{\mathfrak f}$ in Assumption~\ref{ass:Lip_cont_bounded} is equal to $1/a$.
Lipschitz continuity, as per \eqref{eq:cond_init_dens} and \eqref{eq:cond_kern},
holds for the constants $\lambda_0 = 0$ and $\lambda_{\mathfrak f} = 1/\left(\sigma^2\sqrt{2\pi e}\right)$. 
\hfill \qed

\section{State-Space Truncation Procedure}
\label{sec:trunc}

We limit the support of the density functions $\pi_0,t_{\mathfrak s}$ to the sets $\Lambda_0,\varGamma$ respectively, 
and recursively compute support sets $\Lambda_t$, as in \eqref{eq:recur_supp}, 
that are associated to the density functions $\pi_t$. 
Then we employ the quantities $\epsilon,\varepsilon_0$ in Assumption~\ref{ass:marg_conv} 
to compute bounds $\varepsilon_t$, as in \eqref{eq:trunc_error}, 
on the error incurring in disregarding the value of the density functions $\pi_t$ outside the sets $\Lambda_t$. 
Finally we truncate the original, unbounded state space to the set $\Upsilon = \cup_{t=0}^N\Lambda_t$.  

As intuitive, the error related to the spatial truncation depends on the behaviour of the conditional density function $t_{\mathfrak s}$ over the eliminated regions of the state space. 
Suppose that sets $\varGamma,\Lambda_0$ are selected such that Assumption~\ref{ass:marg_conv} is satisfied with constants $\epsilon,\varepsilon_0$:  
then Theorem~\ref{thm:error_trunc} provides an upper bound on the error obtained from manipulating the density functions in time $\pi_t(\cdot)$ exclusively over the truncated regions of the state space. 
\begin{thm}
\label{thm:bound_behaviour}
Under Assumption~\ref{ass:marg_conv}
the functions $\pi_t$ satisfy the bound
\begin{equation*}
0\le \pi_t(s)\le \varepsilon_t, \quad \forall s\in\mathcal S\backslash\Lambda_t,
\end{equation*}
where the quantities $\{\varepsilon_t,\,t\in\mathbb Z_N\}$ are defined recursively by
\begin{equation}
\label{eq:trunc_error}
\varepsilon_{t+1} = \epsilon+M_{\mathfrak f}\varepsilon_t, 
\end{equation}
whereas the support sets $\{\Lambda_t,\,t\in\mathbb Z_N\}$ are computed as 
\begin{equation}
\label{eq:recur_supp}
\Lambda_{t+1} = \Pi_{\bar s}\left(\varGamma\cap(\Lambda_t\times\mathcal S)\right),
\end{equation}
where $\Pi_{\bar s}$ denotes the projection map along the second set of coordinates\footnote{Recall that both $\varGamma$ and $\Lambda\times \mathcal S$ are defined over $\mathcal S^2 = \mathcal S\times\mathcal S$.}.
\end{thm}

\begin{rem}
Notice that if the shape of the sets $\varGamma$ and $\Lambda_0$ is computationally manageable (\eg, if these sets are polytopes),  
then it is possible to precisely implement the computation of the recursion in \eqref{eq:recur_supp} by available software tools, 
such as the MPT toolbox \cite{mpt}. 
Further, 
if for some $t_0$, $\Lambda_{t_0+1}\supset\Lambda_{t_0}$, then for all $t\ge t_0$, $\Lambda_{t+1}\supset\Lambda_t$. 
Similarly, we have that 
\begin{itemize}
\item if for some $t_0$, $\Lambda_{t_0+1}\subset\Lambda_{t_0}$, then for all $t\ge t_0$, $\Lambda_{t+1}\subset\Lambda_t$.
\item if for some $t_0$, $\Lambda_{t_0+1} = \Lambda_{t_0}$, then for all $t\ge t_0$, $\Lambda_{t}=\Lambda_{t_0}$. 
\end{itemize}  
In order to clarify the role of $\varGamma$ in the computation of $\Lambda_t$, 
we emphasize that $\Lambda_{t+1} = \cup_{s\in\Lambda_t}\varXi(s)$, 
where $\varXi$ depends only on $\varGamma$ and is defined by the set-valued map
\begin{equation*}
\varXi:\mathcal S\rightarrow 2^{\mathcal S},\quad \varXi(s) = \{\bar s\in\mathcal S| (s,\bar s)\in\varGamma\}.
\end{equation*}
Figure~\ref{fig:graph_Lambda} provides a visual illustration of the recursion in \eqref{eq:recur_supp}. \qed
\end{rem}

Let us introduce the quantity $\kappa(t,M_{\mathfrak f})$, which plays a role in the solution of \eqref{eq:trunc_error} and will be frequently used shortly:
\begin{equation}
\label{eq:kappa}
\kappa(t,M_{\mathfrak f}) = \left\{
\begin{array}{lll}
\frac{1-M_{\mathfrak f}^t}{1-M_{\mathfrak f}}, & & M_{\mathfrak f}\ne 1\\
t, & & M_{\mathfrak f} = 1.
\end{array}
\right.
\end{equation}

The following theorem provides a truncation procedure, 
valid over a finite time horizon $\mathbb Z_N = \{0,1,\ldots,N\}$, 
which reduces the state space $\mathcal S$ to the set $\Upsilon=\bigcup_{t=0}^{N}\Lambda_t$. 
The theorem also formally quantifies the associated truncation error. 

\begin{thm}
\label{thm:error_trunc}
Suppose that the state space of the process $\mathscr M_{\mathfrak s}$ has been truncated to the set $\Upsilon=\bigcup_{t=0}^{N}\Lambda_t$. 
Let us introduce the following recursion to compute functions $\mu_t:\mathcal S\rightarrow\mathbb R^{\ge 0}$ as an approximation of the density functions $\pi_t$: 
\begin{equation}
\label{eq:rec_trunc}
\mu_{t+1}(\bar s) = \mathds 1_{\Upsilon}(\bar s)\int_{\mathcal S}t_{\mathfrak s}(\bar s|s)\mu_t(s)ds,
\quad \mu_0(s) = \mathds 1_{\Lambda_0}(s)\pi_0(s), 
\quad\forall \bar s\in\mathcal S.
\end{equation}
Then the introduced approximation error is
$\|\pi_t-\mu_t\|_{\infty}\le \varepsilon_t$,
for all $t\in\mathbb Z_N$.
\end{thm}
Recapitulating,  
Theorem~\ref{thm:error_trunc} leads to the following procedure to approximate the density functions $\pi_t$ of $\mathscr M_{\mathfrak s}$ over an unbounded state space $\mathcal S$: 
\begin{enumerate}
\item truncate $\pi_0$ in such a way that $\mu_0$ has a bounded support $\Lambda_0$; 
\item truncate the conditional density function $t_{\mathfrak s}(\cdot |s)$ over
a bounded set for all $s\in\mathcal S$, 
then quantify $\varGamma\subset\mathcal S^2$ as the support of the truncated density function;
\item leverage the recursion in \eqref{eq:recur_supp} to compute the support sets $\Lambda_t$; 
\item use the recursion in \eqref{eq:rec_trunc} to compute the approximate density functions $\mu_{t}$ over the set $\Upsilon = \cup_{t=0}^{N}\Lambda_t$.
Note that the recursion in \eqref{eq:rec_trunc} is effectively computed over the set $\Upsilon$, since $\mu_t(s) = 0$ for all $s\in\mathcal S\backslash\Upsilon$.  
\end{enumerate}

\noindent Note that we could as well handle the support of $\mu_t(\cdot)$ over the time-varying sets $\Lambda_t$, 
by adapting the recursion in \eqref{eq:rec_trunc} with $\mathds 1_{\Lambda_{t+1}}$ instead of $\mathds 1_{\Upsilon}$.  
However, 
while employing the (larger) set $\Upsilon$ may lead to an increase in memory requirements at each stage,  
it will considerably simplify the computations of the state-space partitioning and of the Markov chain abstraction:   
indeed, employing time-varying sets $\Lambda_t$ would render the partitioning procedure also time-dependent, 
and the obtained Markov chain would be time-inhomogeneous.
We therefore opt to work directly with set $\Upsilon$ in order to avoid these difficulties.  

\medskip 

\noindent\textbf{Example \ref{ex:linear_1d} (Continued).}
We easily obtain a closed form for the sets $\Lambda_t = [\beta_t,\gamma_t]$, via 
\begin{equation*}
\beta_{t+1} = a\beta_t+b-\alpha\sigma,\quad
\gamma_{t+1} = a\gamma_t+b+\alpha\sigma.
\end{equation*}
Set $\Upsilon$ is the union of intervals $[\beta_t,\gamma_t]$. 
The error of the state-space truncation over $\Upsilon$ is 
\begin{equation*}
\|\pi_t-\mu_t\|_\infty
\le \varepsilon_t = \kappa(t,M_{\mathfrak f})\frac{\phi_1(\alpha)}{\sigma},\quad M_{\mathfrak f} = \frac{1}{a}.
\end{equation*}
\hfill \qed

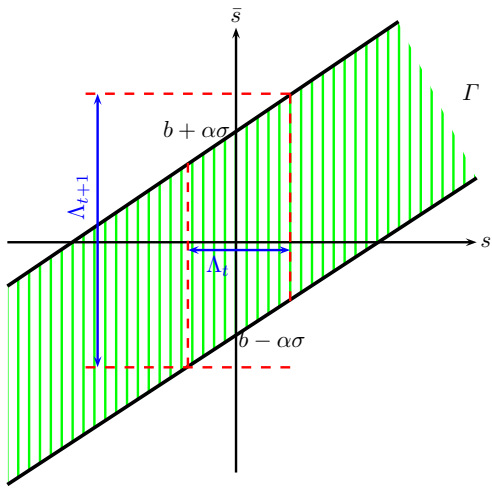
\begin{figure}
\centering
\scalebox{0.8}
{
\begin{pspicture}(-0.5,-3.6)(8,3.6)
\pspolygon[linewidth=0.0020,linecolor=white,fillstyle=vlines,hatchwidth=0.04,hatchangle=0.0,hatchcolor=green](0,-0.6)(6.5,3.8)(7.8,1.2)(0,-3.9)
\psline[linewidth=0.04cm,arrowsize=0.05291667cm 2.0,arrowlength=1.4,arrowinset=0.4]{->}(0,0.13)(7.8,0.13)
\psline[linewidth=0.04cm,arrowsize=0.05291667cm 2.0,arrowlength=1.4,arrowinset=0.4]{->}(3.8,-3.7)(3.8,3.7)
\rput(3.8,3.9){$\bar s$}
\rput(7.95,0.13){$s$}
\psline[linewidth=0.06cm](0,-0.6)(6.5,3.8)
\psline[linewidth=0.06cm](0,-3.9)(7.8,1.2)
\rput(4.4,-1.5){$b-\alpha\sigma$}
\rput(3.15,2){$b+\alpha\sigma$}
\rput(7.7,2.63){$\varGamma$}
\rput(3.5,-0.27){$\textcolor{blue}{\Lambda_t}$}
\psline[linewidth=0.04cm,linecolor=red,linestyle=dashed,dash=0.16cm 0.16cm](3,-1.95)(3,1.45)
\psline[linewidth=0.04cm,linecolor=red,linestyle=dashed,dash=0.16cm 0.16cm](4.7,-0.85)(4.7,2.6)
\psline[linewidth=0.04cm,linecolor=red,linestyle=dashed,dash=0.16cm 0.16cm](4.7,2.6)(1.3,2.6)
\psline[linewidth=0.04cm,linecolor=red,linestyle=dashed,dash=0.16cm 0.16cm](4.7,-1.95)(1.3,-1.95)
\psline[linewidth=0.04cm,linecolor=blue,arrowsize=0.05291667cm 2.0,arrowlength=1.4,arrowinset=0.4]{<->}(1.5,-1.95)(1.5,2.6)
\rput{90.0}(1.9,-0.3){\rput(1.1914062,0.71){\textcolor{blue}{$\Lambda_{t+1}$}}}
\psline[linewidth=0.04cm,linecolor=blue,arrowsize=0.05291667cm 2.0,arrowlength=1.4,arrowinset=0.4]{<->}(3,0)(4.7,0)
\end{pspicture} 
}
\caption{Graphical representation of the recursion in \eqref{eq:recur_supp} for sets $\Lambda_t$.}
\label{fig:graph_Lambda}
\end{figure}

The recursion in \eqref{eq:rec_trunc} indicates that the support of functions $\mu_t$ are always contained in the set $\Upsilon$ (namely, they are equal to zero over the complement of $\Upsilon$). 
We are thus only interested in computing these functions over the set $\Upsilon$,  
which allows
simplifying the recursion in \eqref{eq:rec_trunc} as follows:  
\begin{equation}
\label{eq:recursive_RA2}
\mu_{t+1}(\bar s)=\int_{\Upsilon}t_{\mathfrak s}(\bar s|s)\mu_{t}(s)ds,\quad\forall \bar s\in \Upsilon,\quad t \in \mathbb Z_{N-1}.
\end{equation}

\section{Markov Chain abstraction via State-Space Partitioning}
\label{sec:partition}
In this section we assume that sets $\varGamma,\Lambda_0$ have been properly selected so that $\Upsilon = \cup_{t=0}^{N}\Lambda_t$ is bounded.
In order to formally abstract process $\mathscr M_{\mathfrak s}$ as a finite Markov chain $\mathscr M_{\mathfrak p}$ and to approximate its state density functions,
we select a finite partition of the bounded set $\Upsilon$ as $\Upsilon = \cup_{i=1}^n \mathcal A_i$, 
where sets $\mathcal A_i$ have non-trivial measure.   
We then complete the partition over the whole state space $\mathcal S = \cup_{i=1}^{n+1}\mathcal A_i$ 
by additionally including set $\mathcal A_{n+1} = \mathcal S\backslash\Upsilon$.  
This results in a finite Markov chain $\mathscr M_{\mathfrak p}$ with $n+1$ discrete abstract states in the set $\mathbb N_{n+1} = \{1,2,\ldots,n,n+1\}$, 
and characterised by the transition probability matrix $P = [P_{ij}]\in\mathbb R^{(n+1)\times(n+1)}$,
where the probability of jumping from any pair of states $i$ to $j$ ($P_{ij}$) is computed as 
\begin{equation}
\label{eq:trans_elem}
\begin{array}{l}
P_{ij} = \frac{1}{\mathcal L(\mathcal A_i)}\int_{\mathcal A_j}\int_{\mathcal A_i}t_{\mathfrak s}(\bar s|s)ds d\bar s,\quad\forall i\in\mathbb N_n,\\ 
P_{(n+1)j} = \delta_{(n+1)j}, 
\end{array}
\end{equation}
for all $j\in\mathbb N_{n+1}$, 
and where $\delta_{(n+1)j}$ is the Kronecker delta function (the abstract state $n+1$ of $\mathscr M_{\mathfrak p}$ is absorbing), 
and $\mathcal L(\cdot)$ denotes the Lebesgue measure of a set (i.e., its volume).    
The quantities in \eqref{eq:trans_elem} are well-defined since the set $\Upsilon$ is bounded and the measures $\mathcal L(\mathcal A_i), i\in \mathbb N_n$, 
are finite and non-trivial. 
Notice that matrix $P$ for the Markov chain $\mathscr M_{\mathfrak p}$ is stochastic, namely
\begin{align*}
\sum_{j=1}^{n+1} P_{ij}
& = \sum_{j=1}^{n+1}\frac{1}{\mathcal L(\mathcal A_i)}\int_{\mathcal A_j}\int_{\mathcal A_i}t_{\mathfrak s}(\bar s|s)ds d\bar s
= \frac{1}{\mathcal L(\mathcal A_i)}\int_{\mathcal A_i}\left(\sum_{j=1}^{n+1}\int_{\mathcal A_j}t_{\mathfrak s}(\bar s|s)d\bar s\right)ds\\
& = \frac{1}{\mathcal L(\mathcal A_i)}\int_{\mathcal A_i}\int_{\mathcal S}t_{\mathfrak s}(\bar s|s)d\bar s ds
= \frac{1}{\mathcal L(\mathcal A_i)}\int_{\mathcal A_i}ds
= 1.
\end{align*}
The initial distribution of $\mathscr M_{\mathfrak p}$ is the pmf $\mathbf{p_0} = [p_0(1),p_0(2),\ldots,p_0(n+1)]$, 
and it is obtained from $\pi_0$ as
$p_0(i) = \int_{\mathcal A_i}\pi_0(s)ds, \forall i\in\mathbb N_{n+1}$.
Then the pmf associated to the state distribution of $\mathscr M_{\mathfrak p}$ at time $t$ can be computed as $\mathbf{p_t} = \mathbf{p_0} P^t$.

It is intuitive that the discrete pmf $\mathbf{p_t}$ of the Markov chain $\mathscr M_{\mathfrak p}$ approximates the continuous density function $\pi_t$ of the Markov process $\mathscr M_{\mathfrak s}$. 
In the rest of the section we show how to formalise this relationship: 
$\mathbf{p_t}$ is used to construct an approximate function, denoted by $\psi_t$, 
of the density $\pi_t$. 
Theorem~\ref{thm:Error} shows that $\psi_t$ is a piecewise constant approximation 
(with values in its codomain that are the entries of the pmf $\mathbf{p_t}$, normalised by the Lebesgue measure of the associated partition set) 
of the original density function $\pi_t$. 
Moreover, 
under the continuity assumption in \eqref{eq:cond_kern} (ref. Lemma~\ref{lmm:lip_cont}) we can establish the Lipschitz continuity of $\pi_t$, 
which enables the quantification (in Theorem~\ref{thm:Error}) of the error related to its piecewise constant approximation $\psi_t$.    

\begin{lem}
\label{lmm:lip_cont}
Suppose that the inequality in \eqref{eq:cond_kern} holds. 
Then the state density functions $\pi_t(\cdot)$ are globally Lipschitz continuous with constant $\lambda_{\mathfrak f}$ for all $t\in\mathbb N$: 
\begin{equation*}
|\pi_t(s)-\pi_t(s')|\le \lambda_{\mathfrak f}\|s-s'\|,\quad\forall s,s'\in\mathcal S. 
\end{equation*}
\end{lem}
\begin{thm}
\label{thm:Error}
Under Assumptions~\ref{ass:marg_conv} and \ref{ass:Lip_cont_bounded}, 
the functions $\pi_t(\cdot)$ can be approximated by piecewise constant functions $\psi_t(\cdot)$, defined as
\begin{equation}
\label{eq:approx_marg}
\psi_t(s) = \sum_{i=1}^{n}\frac{p_t(i)}{\mathcal L(\mathcal A_i)}\mathds 1_{\mathcal A_i}(s),\quad \forall t\in\mathbb N,
\end{equation}
where $\mathds 1_{B}(\cdot)$ is the indicator function of a set $B\subset\mathcal S$.
The approximation error is upper-bounded by the quantity
\begin{equation}
\label{eq:error}
\|\pi_t-\psi_t\|_\infty \le\varepsilon_t + E_t,\quad \forall t\in\mathbb N,
\end{equation}
where $E_t$ can be recursively computed as 
\begin{equation}
\label{eq:part_error_rec}
E_{t+1} = M_{\mathfrak f} E_t+\lambda_{\mathfrak f}\delta,\quad E_0 = \lambda_0\delta,
\end{equation}
and $\delta$ is an upper bound on the diameters of the partition sets $\{\mathcal A_i\}_{i=1}^{n}$, namely
$\delta\, = \sup\left\{\|s-s'\|,\,\, s,s'\in\mathcal A_i,\,\, i\in\mathbb N_n\right\}$.
\end{thm}
Note that the functions $\psi_t$ are defined over the whole state space $\mathcal S$, 
but \eqref{eq:approx_marg} implies that they are equal to zero outside the set $\Upsilon$.
\begin{cor}
The recursion in \eqref{eq:part_error_rec} admits the explicit solution 
\begin{equation*}
E_t = \left[\kappa(t,M_{\mathfrak f})\lambda_{\mathfrak f}+M_{\mathfrak f}^t\lambda_0\right]\delta,
\end{equation*}
where $\kappa(t,M_{\mathfrak f})$ is introduced in \eqref{eq:kappa}.  
\end{cor}

Underlying Theorem~\ref{thm:Error} is the fact that $\psi_t(\cdot)$ are in general sub-stochastic density functions:  
\begin{align*}
\int_{\mathcal S}\psi_t(s)ds 
& = \int_{\mathcal S}\sum_{i=1}^{n}\frac{p_t(i)}{\mathcal L(\mathcal A_i)}\mathds 1_{\mathcal A_i}(s)ds
= \sum_{i=1}^{n}\frac{p_t(i)}{\mathcal L(\mathcal A_i)}\int_{\mathcal S}\mathds 1_{\mathcal A_i}(s)ds\\
& = \sum_{i=1}^{n}\frac{p_t(i)}{\mathcal L(\mathcal A_i)}\mathcal L(\mathcal A_i) = \sum_{i=1}^{n}p_t(i)=1-p_t(n+1)\le 1.
\end{align*} 
This is clearly due to the fact that we are operating on the dynamics of $\mathscr M_{\mathfrak s}$ truncated over the set $\Upsilon$.  
It is thus intuitive that the approximation procedure and the derived error bounds are also valid for the case of sub-stochastic density functions \cite{AKNP14}, 
namely 
\begin{equation*}
\int_{\mathcal S}t_{\mathfrak s}(\bar s|s)d\bar s\le 1,\quad\forall s\in\mathcal S,\quad \int_{\mathcal S}\pi_0(s)ds\le 1, 
\end{equation*}
the only difference being that the obtained Markov chain $\mathscr M_{\mathfrak p}$ is as well sub-stochastic.  
 
Further, 
whenever the Lipschitz continuity requirement on the initial density function, as per \eqref{eq:cond_init_dens} in Assumption~\ref{ass:Lip_cont_bounded}, does not hold, 
(for instance, this is the case when the initial state of the process is deterministic)
we can relax this continuity assumption on the initial distribution of the process by starting the discrete computation from the time step $t=1$. 
In this case we define the pmf $\mathbf{p_1}= [p_1(1),p_1(2),\ldots,p_1(n+1)]$,
where
\begin{equation}
\label{eq:relax_cont}
p_1(i) = \int_{\mathcal A_i}\int_{\mathcal S}t_{\mathfrak s}(\bar s|s)\pi_0(s)ds d\bar s, \quad \forall i\in\mathbb N_{n+1},
\end{equation}
and derive $\mathbf{p_t} = \mathbf p_1 P^{t-1}$ for all $t\in\mathbb N$. 
Theorem~\ref{thm:Error} follows along similar lines, 
except for equation \eqref{eq:part_error_rec}, where the initial error is set to $E_0 = 0$ 
and the time-dependent terms $E_t$ can be derived as 
$E_t = \kappa(t,M_{\mathfrak f})\lambda_{\mathfrak f}\delta$.

It is important to emphasise the practical computability of the derived errors, 
and the fact that they can be tuned by selecting a finer partition of set $\Upsilon$ that relates to a smaller global parameter $\delta$. 
Further,
in order to attain abstractions that are practically useful, it imperative to seek improvements on the derived error bounds: 
in particular, the approximation errors can be computed locally (under corresponding local Lipschitz continuity assumptions), 
following the procedures discussed in \cite{SA13}.

\medskip

\noindent\textbf{Example \ref{ex:linear_1d} (Continued).}
The error of the Markov chain abstraction can be expressed as
\begin{equation}
\label{eq:error_example_1}
\|\pi_t-\psi_t\|_\infty
\le \kappa(t,M_{\mathfrak f})\left[\frac{\delta}{\sigma^2\sqrt{2\pi e}}+\frac{\phi_1(\alpha)}{\sigma}\right]
,\quad M_{\mathfrak f} = \frac{1}{a}.
\end{equation}
The error can be tuned in two distinct ways: 
\begin{enumerate}
\item 
by selecting larger values for $\alpha$, 
which on the one hand leads to a less narrow truncation, 
but on the other requires the partition of a larger interval;  
\item 
by reducing partitions diameter $\delta$, which of course results in a larger cardinality of the partition sets.   
\end{enumerate}
Let us select values $b = 0, \beta_0 = 0,\gamma_0 = 1,\sigma = 0.1$, and time horizon $N=5$. 
For $a=1.2$ we need to partition the interval $\Upsilon = \left[-0.75\alpha,2.49+0.75\alpha\right]$, 
which results in the error $\|\pi_t-\psi_t\|_\infty\le 86.8\delta+35.9\phi_1(\alpha)$ for all $t\in\mathbb Z_N$.
For $a = 0.8$ we need to partition the smaller interval
$\Upsilon = \left[-0.34\alpha,0.33+0.34\alpha\right]$, 
which results in the error $\|\pi_t-\psi_t\|_\infty\le 198.6\delta+82.1\phi_1(\alpha)$ for all $t\in\mathbb Z_N$. 
Notice that in the case of $a = 1.2$, 
we partition a larger interval and obtain a smaller error, 
while for $a=0.8$ we partition a smaller interval with correspondingly a larger error. 
It is obvious that the parameters $\delta,\alpha$ can be chosen properly to ensure that a certain error precision is met.
This simple model admits a solution in closed form, 
and its state density functions can be obtained as the convolution of a uniform distribution (the contribution of initial state) and a zero-mean Gaussian distribution with time-dependent variance (the contributions of the process noise). 
This leads to the plots in Figure~\ref{fig:1d_density}, 
which display the original and the approximated state density functions for the set of parameters $\alpha = 2.4,\delta = 0.05$. 
\hfill \qed
\begin{figure}
\centering
\includegraphics[scale = 0.5]{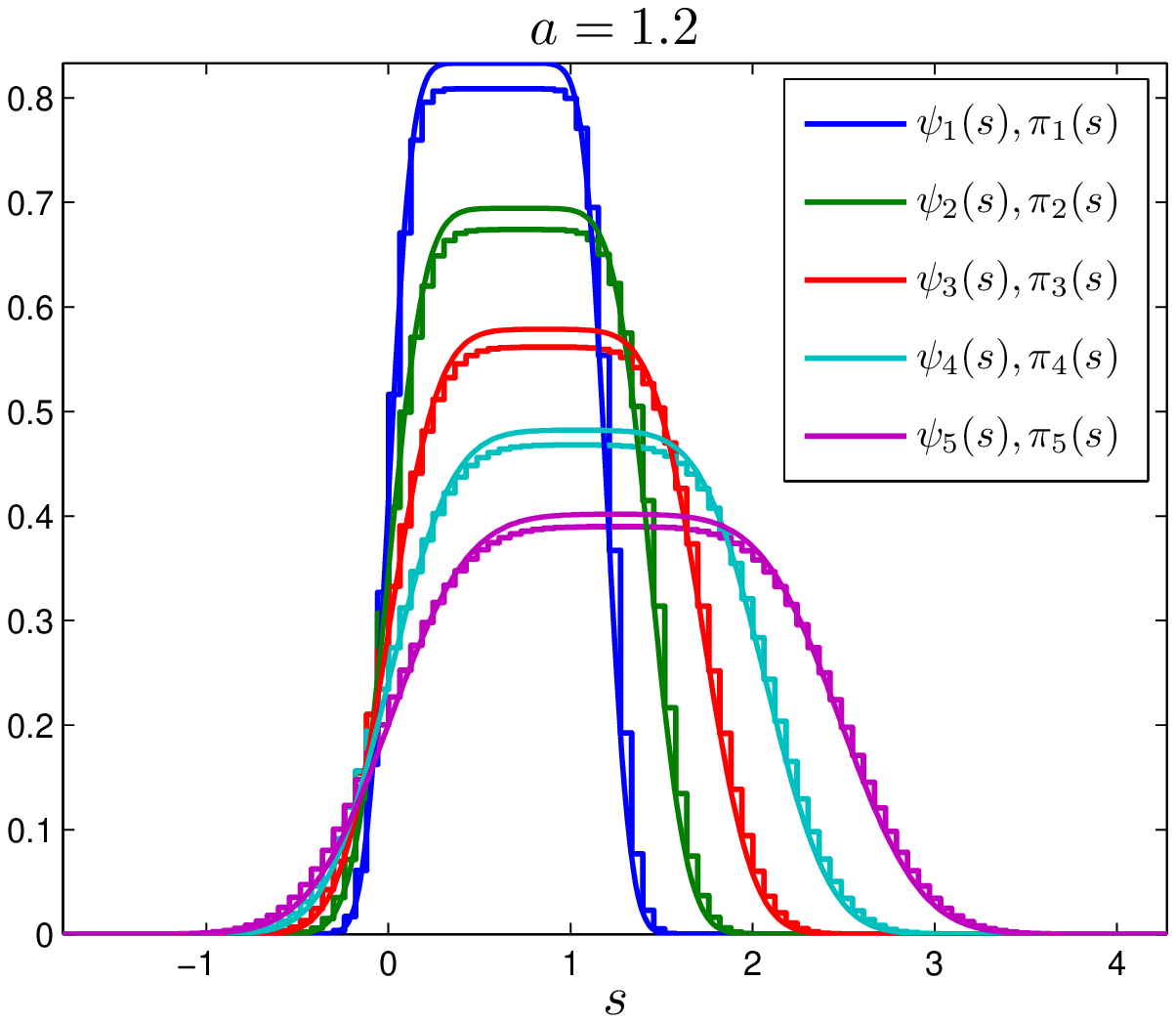}
\includegraphics[scale = 0.5]{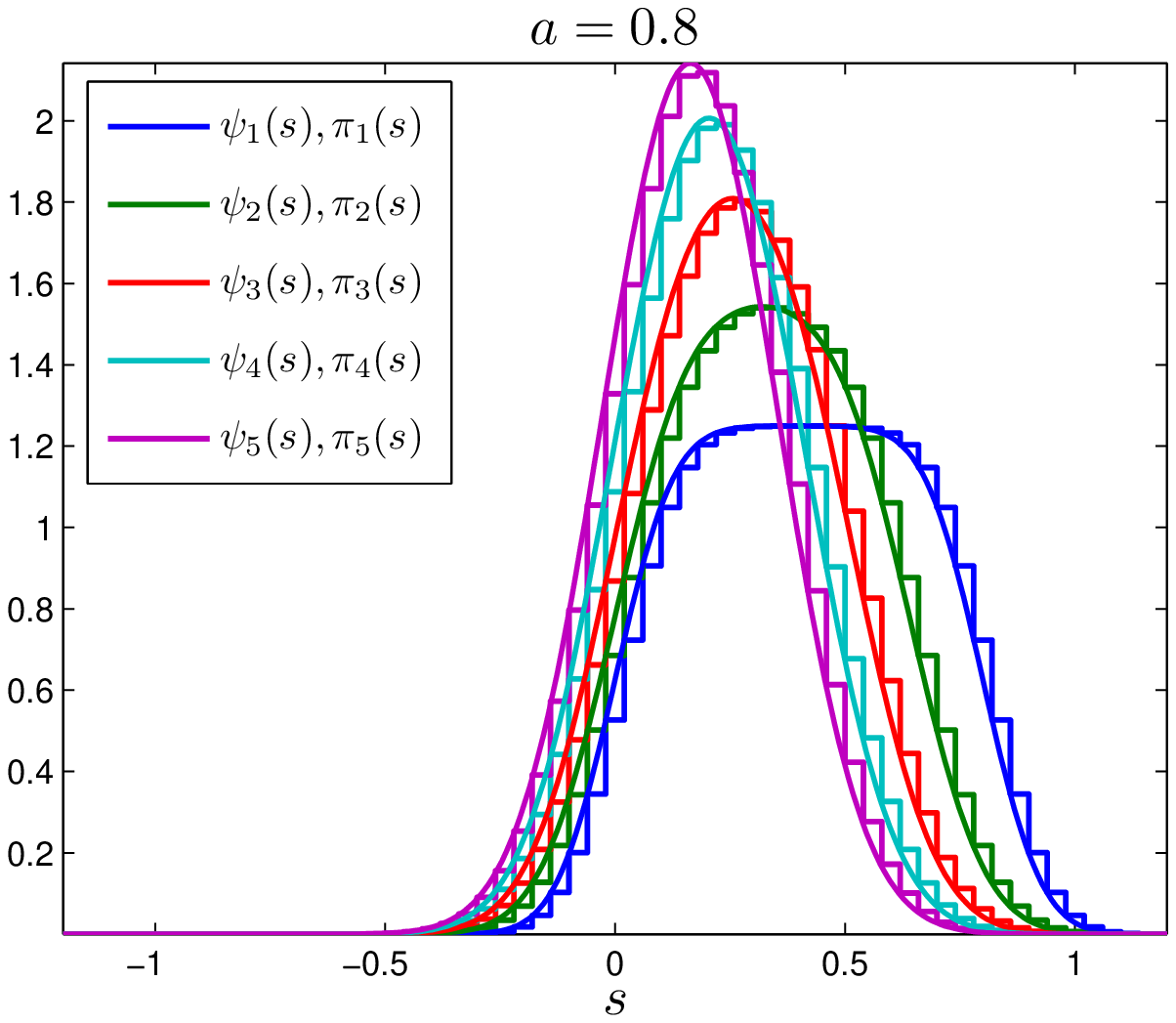}
\caption{Piecewise constant approximation $\psi_t(\cdot)$ of the state density function $\pi_t(\cdot)$ (derived analytically), 
for parameters $a=1.2$ (left) and $a=0.8$ (right).} 
\label{fig:1d_density}
\end{figure}

\section{Higher-Order Approximation Schemes}
\label{sec:approx&error}
In the previous section we have shown that a Markov chain abstraction can be employed to formally approximate the density function of a Markov process in time. 
This abstraction is interpreted as a piecewise-constant approximation of the density function of the Markov model. 
In this section we argue that this procedure can be extended to approximation techniques based on higher-order interpolations. 

With focus on the truncated region of the state space, 
let us denote with $\mathbb B(\Upsilon)$ the space of bounded and measurable functions $f:\Upsilon\rightarrow\mathbb R$,
equipped with the infinity norm
$\Vert f\Vert_{\infty} = \sup\{|f(s)|, \, s\in\Upsilon\}$, 
for all $f\in \mathbb B(\Upsilon)$. 
The linear operator $\mathcal R_\Upsilon$, defined over $\mathbb B(\Upsilon)$ by 
\begin{equation}\label{eq:operator}
\mathcal R_\Upsilon (f)(\bar s) = \int_{\Upsilon}t_{\mathfrak s}(\bar s|s)f(s)ds, 
\quad \forall \bar s\in\Upsilon,
\quad \forall f\in\mathbb B(\Upsilon),
\end{equation}
characterises the solution of the recursion in \eqref{eq:recursive_RA2} as
$\mu_t(s) = \mathcal R_\Upsilon^t(\mu_0) (s)$, for any $t\in\mathbb N_N$.
While in Section~\ref{sec:partition} we have proposed approximations of functions $\mu_t(\cdot)$ by piecewise-constant 
functions $\psi_t(\cdot)$ with an explicit quantification of the associated error, 
in this section we are interested in considering approximations via higher-order interpolations.

\subsection{Quantification of the Error of a Projection Over a Function Space}

Consider a set of basis functions $\Phi = \{\phi_1(s),\phi_2(s),\ldots,\phi_h(s)\}$, $h\in\mathbb N$,
the function space $\Psi = span\,\Phi$ generated by this set as a subset of $\mathbb B(\Upsilon)$, 
and a linear operator $\Pi_{\Upsilon}:\mathbb B(\Upsilon)\rightarrow \Psi$, 
which projects any function $f\in\mathbb B(\Upsilon)$ onto the function space $\Psi$.
Theorem~\ref{thm:error_dynamic} provides a theoretical result for approximating the solution of \eqref{eq:recursive_RA2}: 
the following section provides details on turning this result into a useful tool for approximations.
\begin{thm}\label{thm:error_dynamic}
Assume that a linear projection operator $\Pi_{\Upsilon}:\mathbb B(\Upsilon)\rightarrow \Psi$ satisfies the inequality
\begin{equation}\label{eq:Gener_opr}
\left\| \Pi_{\Upsilon}(t_{\mathfrak s}(\cdot|s)) - t_{\mathfrak s}(\cdot|s)\right\|_\infty\le\mathcal E^{\mathfrak h},  
\quad \forall s\in\Upsilon,
\end{equation}
and that there exists a finite constant $M_{\mathfrak f}^{\mathfrak h}$, 
such that
\begin{equation}\label{eq:constant_M}
\int_{\Upsilon}\left|\Pi_{\Upsilon}(t_{\mathfrak s}(\bar s|s))\right| ds\le M_{\mathfrak f}^{\mathfrak h}, 
\quad \forall \bar s\in\Upsilon.
\end{equation}
Define the functions $\psi_t^{\mathfrak h}(\cdot)$ as approximations of $\mu_t(\cdot)$ (cf. \,\eqref{eq:operator}), by
\begin{equation}\label{eq:dis_val_fun}
\psi_t^{\mathfrak h} = (\Pi_{\Upsilon}\mathcal R_\Upsilon)^t(\mu_0), \quad t\in\mathbb Z_N.
\end{equation}
Then it holds that 
\begin{equation}\label{eq:Main_inequality}
\|\mu_t-\psi_t^{\mathfrak h}\|_\infty\le E_t^{\mathfrak h}, \quad t\in\mathbb N_N,
\end{equation}
where the error $E_t^{\mathfrak h}$ satisfies the difference equation
\begin{equation*}
E_{t+1}^{\mathfrak h} = M_{\mathfrak f}^{\mathfrak h} E_t^{\mathfrak h} + \mathcal E^{\mathfrak h},\quad E_0^{\mathfrak h} = 0. 
\end{equation*}
\end{thm}
\begin{cor}
\label{thm:dir_error}
Under the assumptions raised in \eqref{eq:Gener_opr}-\eqref{eq:constant_M}, 
the error $E_t^{\mathfrak h}$ can be alternatively expressed explicitly as $E_t^{\mathfrak h} = \mathcal E^{\mathfrak h}\kappa(t,M_{\mathfrak f}^{\mathfrak h})$.
\end{cor}
The error $E_t^{\mathfrak h}$ formulated in Theorem \ref{thm:error_dynamic} is comparable with the quantity $E_t$ computed in Theorem \ref{thm:Error}. 
Both $E_t^{\mathfrak h},E_t$ represent bounds on the approximation error introduced by $\mu_t(\cdot)$, 
the density function obtained after state space truncation. 
The difference is in the initialisation of the corresponding recursions, 
where we have $E_0^{\mathfrak h} = 0$ because $\psi_0^{\mathfrak h} = \mu_0$, 
but $E_0 = \lambda_0\delta$ since $\psi_0$ is a piecewise constant approximation of $\mu_0$ in \eqref{eq:approx_marg}. 
As we mentioned before, the quantities in \eqref{eq:relax_cont} can be alternatively employed as the starting values of the computation to relax the continuity assumption on $\pi_0$, which results in an initial error $E_0 = 0$, thus providing a complete similarity between $E_t$ and $E_t^{\mathfrak h}$.

\subsection{Construction of the Projection Operator}

In the ensuing sections we focus, for the sake of simplicity, on a Euclidean domain,   
namely $\Upsilon\subset \mathcal S=\mathbb{R}^d$, 
where $d$ denotes a finite dimension.
We discuss a general form for an interpolation operator related to the discussed projection operation. 
Let $\phi_j:\mathcal D\subset\mathbb R^d\rightarrow\mathbb R, j\in\mathbb N_h,$ be independent functions 
defined over a generic set $\mathcal D$. 
The interpolation operator $\Pi_{\mathcal D}$ is defined
as a projection map onto the function space $\Psi = span\{\phi_1(s),\phi_2(s),\ldots,\phi_h(s)\}$, 
which projects any function $f:\mathcal D\rightarrow \mathbb R$
to a unique function $\Pi_{\mathcal D}(f) = \sum_{j=1}^{h}\alpha_j\phi_j$,
using a finite set of data $\{(s_j,f(s_j))| s_j\in\mathcal D, j\in\mathbb N_h\}$ and 
such that $\Pi_{\mathcal D}(f)(s_j) = f(s_j)$.  
The projection coefficients $\alpha_j,j\in\mathbb N_h,$ satisfy the linear equation $\mathbf f = \mathcal Q \boldsymbol{\alpha}$,
where $\mathbf f = [f(s_i)]_{i \in\mathbb N_h}$ and $\boldsymbol{\alpha} = [\alpha_j]_{j \in\mathbb N_h}$ are $h$--dimensional column vectors, 
and $\mathcal Q = [\phi_j(s_i)]_{i,j}$ is the associated $(h\times h)$--dimensional interpolation matrix.

\medskip

Let us now shift the focus to the recursion in \eqref{eq:recursive_RA2} discussed in the previous section and tailor the operators above accordingly. 
Let us select a partition $\{\mathcal A_i,\,i\in\mathbb N_n\}$ for the set $\Upsilon$, 
with finite cardinality $n$. 
Selecting a basis $\{\phi_{ij},\,j\in\mathbb N_h\}$ for each partition, 
let us introduce the interpolation operators $\Pi_{\mathcal A_i}$ 
for the projection over each partition set $\mathcal A_i$,  
which is done as described above by replacing the domain $\mathcal D$ with $\mathcal A_i$. 
Finally, let us introduce the (global) linear operator $\Pi_{\Upsilon}$, 
acting on a function $f:\Upsilon\rightarrow \mathbb{R}$ by 
\begin{equation}
\label{eq:operator_partition}
\Pi_{\Upsilon}(f) = \sum_{i=1}^n \mathds{1}_{\mathcal A_i}\Pi_{\mathcal A_i}(f|_{\mathcal A_i}), 
\end{equation}
where $f|_{\mathcal A_i}$ represents the restriction of the domain of function $f$ to the partition set $\mathcal A_i$.

\subsection{Approximation Algorithm}
An advantage of the interpolation operator in \eqref{eq:operator_partition} is that $\Pi_{\Upsilon}(f)$ 
is fully characterised by the interpolation coefficients $\alpha_{ij}$, since 
\begin{equation*}
\Pi_{\Upsilon}(f) = \sum_{i=1}^{n}\sum_{j=1}^{h}\alpha_{ij}\phi_{ij}\mathds{1}_{\mathcal A_i}.
\end{equation*} 
The set of interpolation coefficients $\alpha_{ij}$ is computable by matrix multiplication 
based on the data set $\{f(s_{ij}),\,i\in\mathbb N_n,j\in\mathbb N_h\}$.
More precisely, we have $[f(s_{iv})]_{v\in\mathbb N_h} = \mathcal Q_i [\alpha_{ij}]_{j\in\mathbb N_h}$ with the interpolation matrices $\mathcal Q_i = [\phi_{ij}(s_{iv})]_{v,j\in\mathbb N_h}$.
These matrices depend solely on the interpolation points $s_{ij}$ and on the basis functions $\phi_{ij}$ evaluated at these points 
and can be computed off-line (see step~\ref{alg:comp} in Algorithm~\ref{algo:app_val}, to be discussed shortly).
Moreover, values of the function $f$ only at the interpolation points $s_{ij}$ are sufficient for the computation of $\alpha_{ij}$. 

Let us now focus on the recursion in \eqref{eq:dis_val_fun}, 
namely $\psi_{t+1}^{\mathfrak h} = \Pi_{\Upsilon}\mathcal R_\Upsilon(\psi_t^{\mathfrak h})$, 
given the initialisation $\psi_0^{\mathfrak h} = \mu_0$, 
for the approximate computation of the value functions. 
This recursion indicates that the approximate functions $\psi_t^{\mathfrak h},\, t\in\mathbb N_N,$ belong to the image of the operator $\Pi_{\Upsilon}$, 
and as such can be expressed as 
\begin{equation*}
\psi^{\mathfrak h}_t = \sum_{i=1}^{n}\sum_{j=1}^{h}\alpha_{ij}^{t}\phi_{ij}\mathds{1}_{\mathcal A_i},
\end{equation*}
where $\alpha_{ij}^t$ denote the interpolation coefficients referring to function $\psi^{\mathfrak h}_t$ (at step $t$).  
This suggests that we need to store and update the coefficients $\alpha_{ij}^{t}$ for each iteration in \eqref{eq:dis_val_fun}.   
Writing the recursion in the form $\psi^{\mathfrak h}_{t+1} = \Pi_{\Upsilon}(\mathcal R_\Upsilon(\psi^{\mathfrak h}_t))$ indicates that the function $\psi^{\mathfrak h}_{t+1}$ is in the range of the projection $\Pi_{\Upsilon}$.
Therefore, it is sufficient to evaluate the function $\mathcal R_\Upsilon(\psi^{\mathfrak h}_t)$ over the interpolation points in order to compute the coefficients $\alpha_{ij}^{t+1}$.  
In the following expressions, 
the pair $i,u$ indicates the indices of related partition sets, 
namely $\mathcal A_i,\mathcal A_u$,
whereas the pair of indices $j,v$ show the ordering positions within partition sets. 
For an arbitrary interpolation point $s_{uv}$ we have:
\begin{align*}
\mathcal R_\Upsilon(\psi^{\mathfrak h}_t)(s_{uv})
= \int_{\Upsilon} t_{\mathfrak s}(s_{uv}|s)\psi^{\mathfrak h}_t(s)ds
 = \sum_{i=1}^{n}\sum_{j=1}^{h}\alpha_{ij}^{t}\int_{\mathcal A_i}t_{\mathfrak s}(s_{st}|s)\phi_{ij}(s)ds.
\end{align*}
Introducing the following quantities
\begin{equation*}
P_{ij}^{uv} = \int_{\mathcal A_i}t_{\mathfrak s}(s_{uv}|s)\phi_{ij}(s)ds,
\end{equation*}
we can succinctly express 
\begin{equation*} 
\mathcal R_\Upsilon(\psi^{\mathfrak h}_t)(s_{uv}) = 
\sum_{i=1}^{n}\sum_{j=1}^{h}\alpha_{ij}^{t}P_{ij}^{uv} \doteq \beta^{t+1}_{uv}.
\end{equation*}
Algorithm~\ref{algo:app_val} provides a general procedure for the discrete computation of the interpolation coefficients and of the approximate value functions.
\begin{algorithm}
\caption{Approximate computation of the functions $\psi^{\mathfrak h}_t$}
\label{algo:app_val}
\begin{center}
\begin{algorithmic}[1]
\REQUIRE 
Density function $t_{\mathfrak s}(\bar s|s)$, 
set $\Upsilon$ 
\STATE
Select a finite $n$-dimensional partition of the set $\Upsilon = \cup_{i=1}^{n}\mathcal A_i$
($\mathcal A_i$ are non-overlapping)
\STATE
For each $\mathcal A_i$, select interpolation basis functions $\phi_{ij}$
and points $s_{ij}\in \mathcal A_i$, where $j\in\mathbb N_h$
\STATE \label{alg:marginals}
Compute $P_{ij}^{uv} = \int_{\mathcal A_i}t_{\mathfrak s}(s_{uv}|s)\phi_{ij}(s)ds$, where $i,u\in\mathbb N_n$ and $j,v\in\mathbb N_h$
\STATE \label{alg:comp}
Compute a matrix representation for the operators $\Pi_{\mathcal A_i}$, namely $\mathcal Q_i = [\phi_{ij}(s_{iv})]_{v,j}$
\STATE \label{alg:init_marginals}
 Set $t=1$ and $\beta^1_{ij} = \int_{\Upsilon}t_{\mathfrak s}(s_{ij}|s)\mu_0(s)ds$, for all $i,j$
\IF{$t\le N$}
\STATE \label{alg:comp_1} 
Compute interpolation coefficients $\alpha_{ij}^t$ based on equation
$\mathcal Q_i [\alpha_{ij}^t]_{j\in\mathbb N_h} = [\beta^t_{iv}]_{v\in\mathbb N_h}$,
given $\beta^t_{ij}$ and matrices $\mathcal Q_i$ in step~\ref{alg:comp}
\STATE \label{alg:comp_2}
Compute values $\beta^{t+1}_{uv}$ as $\beta^{t+1}_{uv} = \sum_i\sum_j\alpha_{ij}^{t}P_{ij}^{uv}$, for all $u,v$
\STATE
$t = t+1$
\ENDIF
\ENSURE
Approximate functions $\psi^{\mathfrak h}_t = \sum_i\sum_j\alpha_{ij}^{t}\phi_{ij}\mathds{1}_{\mathcal A_i},\, t\in\mathbb N_N$
\end{algorithmic}
\end{center}
\end{algorithm}

It is possible to simplify Algorithm~\ref{algo:app_val} when the interpolation matrices $\mathcal Q_i$ are nonsingular.
Let us transform the basis $\{\phi_{ij},\,j\in\mathbb N_h\}$ to its equivalent basis using matrix $\left(\mathcal Q_i^T\right)^{-1}$. The interpolation matrices corresponding to the new basis will be the identity matrix. In other words, the new basis functions admit $\mathcal Q_i = \mathbb I_h$ and thus
step~\ref{alg:comp} can be skipped, 
and that the main update (steps~\ref{alg:comp_1} and \ref{alg:comp_2}) can be simplified as follows: 
\begin{equation*}
\alpha_{uv}^{t+1} = \sum_{i=1}^{n}\sum_{j=1}^{h}\alpha_{ij}^t P_{ij}^{uv}, 
\quad\forall u\in\mathbb N_n,v\in\mathbb N_h.
\end{equation*}

In Algorithm~\ref{algo:app_val}, 
the interpolation points $s_{ij}$ are in general pair-wise distinct. 
By extending the domain of interpolation $\mathcal A_i$ to its closure $\bar{\mathcal A}_i$, 
it is legitimate to use boundary points as interpolation points, 
which can lead to a reduction of the number of integrations required in Algorithm~\ref{algo:app_val}.
In the ensuing sections, we will exploit this feature by specifically selecting equally spaced interpolation points.

\section{Special Forms of the Projection Operator} 
\label{sec:proj}

In this section we leverage known interpolation theorems for the construction of the projection operator $\Pi_{\Upsilon}$: 
this should both yield useful schemes for a number of standard models, 
and further help with the understanding of the details discussed in the previous section.

\subsection{Piecewise Constant Approximations}
\label{subsec:PWC}

We focus on the special case of the approximation of a function by a piecewise constant one, 
which has inspired Section~\ref{sec:partition}.
Let us select the basis functions $\phi_{ij}(s) = 1$ for all $i\in\mathbb N_n,j\in\mathbb N_1$ -- 
the cardinality of these sets of basis functions is simply equal to $h=1$  
(we eliminate the corresponding indices when appropriate).  
In this case the matrix operators $\mathcal Q_i,\,i\in\mathbb N_n,$ (\cf step~\ref{alg:comp} in Algorithm~\ref{algo:app_val}) correspond to the identity matrix, 
and the projection operator $\Pi_{\Upsilon}$ becomes 
\begin{equation}
\label{eq:proj_piecewise_constant}
\Pi_{\Upsilon}(f) = \sum_{i=1}^n f(s_i)\mathds{1}_{\mathcal A_i},\quad \forall f\in\mathbb B(\Upsilon),
\end{equation}
where the quantities $P_{ij}^{uv}$ (\cf step~\ref{alg:marginals} in Algorithm~\ref{algo:app_val}) form a square matrix 
(see step~\ref{alg:marginals_1} in Algorithm~\ref{algo:app_val_abs}).   
The procedure is detailed in Algorithm~\ref{algo:app_val_abs},
while the associated error is formally quantified in Theorem~\ref{thm:zero_abs}.
\begin{algorithm}
\caption{Piecewise constant computation of the functions $\psi^{\mathfrak h}_{t}$}  
\label{algo:app_val_abs}
\begin{center}
\begin{algorithmic}[1]
\REQUIRE 
Density function $t_{\mathfrak s}(\bar s|s)$, set $\Upsilon$
\STATE
Select a finite $n$-dimensional partition of the set $\Upsilon = \cup_{i=1}^{n}\mathcal A_i$
($\mathcal A_i$ are non-overlapping)
\STATE
For each $\mathcal A_i$, select one representative point $s_i\in \mathcal A_i$
\STATE \label{alg:marginals_1}
Compute matrix $P = [P(i,j)]_{i,j}$ with entries $P(i,j) = \int_{\mathcal A_i}t_{\mathfrak s}(s_j|s)ds$, where $i,j\in\mathbb N_n$
\STATE Set $t=1$ and $\alpha_1(i) = \int_{\Upsilon}t_{\mathfrak s}(s_i|s)\mu_0(s)ds$, for all $i$
\IF{$t < N$}
\STATE
Compute the row vector $\boldsymbol{\alpha_{t+1}} = [\alpha_{t+1}(i)]_i$ based on $\boldsymbol{\alpha_{t+1}} = \boldsymbol{\alpha_t}P$
\STATE
$t = t+1$
\ENDIF
\ENSURE
Approximate functions $\psi^{\mathfrak h}_t = \sum_{i=1}^{n}\alpha_t(i)\mathds{1}_{\mathcal A_i},\, t\in\mathbb N_N$
\end{algorithmic}
\end{center}
\end{algorithm}

\begin{thm}
\label{thm:zero_abs} 
Suppose the density function $t_{\mathfrak s}(\cdot|s)$ satisfies the Lipschitz continuity assumption \eqref{eq:cond_kern} with constant $\lambda_{\mathfrak f}$. Then the projection operator \eqref{eq:proj_piecewise_constant} satisfies the inequality 
\begin{equation*}
\|\Pi_{\Upsilon}\left(t_{\mathfrak s}(\cdot|s)\right)-t_{\mathfrak s}(\cdot|s)\|_{\infty} \le \lambda_{\mathfrak f}\delta, \quad \forall s\in\Upsilon,
\end{equation*}
where $\delta = \max_i \delta_i$ is the partition diameter of $\cup_{i=1}^{n}\mathcal A_i = \Upsilon$,
with $\delta_i =\sup\{\|s-s'\|: s,s'\in\mathcal A_i\}$.
Theorem~\ref{thm:error_dynamic} ensures that the approximation error of Algorithm~\ref{algo:app_val_abs} is upper bounded by the quantity
\begin{equation*}
\|\mu_t-\psi_t^{\mathfrak h}\|_\infty\le E_t^{\mathfrak h} = \lambda_{\mathfrak f}\delta\kappa(t,M_{\mathfrak f}), \quad t\in\mathbb N_N,
\end{equation*}
with the constant $M_{\mathfrak f}$ defined in Assumption~\ref{ass:Lip_cont_bounded}. 
\end{thm}
Notice that the error $E_t^{\mathfrak h}$ of Theorem~\ref{thm:zero_abs} reduces to $E_t$ in Theorem~\ref{thm:Error} when employing the quantities in \eqref{eq:relax_cont} to relax the continuity assumption on $\pi_0$. 

Let us compare Algorithms~\ref{algo:app_val} and \ref{algo:app_val_abs} in terms of their computational complexity. 
Algorithm~\ref{algo:app_val} requires $nh(nh+1)$ integrations in the marginalisation steps (\ref{alg:marginals} and \ref{alg:init_marginals}),
whereas $n(n+1)$ integrations are required in Algorithm~\ref{algo:app_val_abs}.  
Furthermore, steps~\ref{alg:comp} and \ref{alg:comp_1} in Algorithm~\ref{algo:app_val} can be skipped by using proper equivalent basis functions, 
whereas these steps are not needed at all in Algorithm~\ref{algo:app_val_abs}. 
As a bottom line, 
higher interpolation orders increase the computational complexity of the approximation procedure, 
however this can as well lead to a lower global approximation error.   
From a different perspective, since the global approximation error depends on the local partitioning sets (their diameter and the local continuity of the density function), 
for a given error higher interpolation procedures may require less partitions sets.

As a final note, comparing the transition probabilities of \eqref{eq:trans_elem} with quantities $P(i,j)$ in step~\ref{alg:marginals_1} of Algorithm~\ref{algo:app_val_abs} reveals that the Markov chain abstraction presented in Section~\ref{sec:partition} is a special case of Algorithm~\ref{algo:app_val_abs}.
More precisely, the \emph{mean value theorem} for integration ensures the existence of representative points $s_i$ such that  $P(i,j)$ of Algorithm~\ref{algo:app_val_abs} is equal to $P_{ij}$ in \eqref{eq:trans_elem}.

\subsection{Higher-order Approximations for One-Dimensional Systems}

We study higher-order interpolations over the real axis, 
where the partition sets $\mathcal A_i$ are real-valued intervals. 
We use this simple setting to quantify the error related to the approximate computation of the functions $\mu_t$.
We select equally spaced points as the interpolation points and employ polynomial basis functions within each interval.

Consider a one dimensional Markov process, $\mathcal S = \mathbb R$, with a partitioning of $\Upsilon = \cup_{i=1}^{n}\mathcal A_i$ which is such that $\mathcal A_i = [a_i,b_i]$. 
Define the interpolation operator $\Pi_\Upsilon$ of \eqref{eq:operator_partition} over the polynomial basis functions $\phi_{ij}(s) = s^{j-1}$, $i\in\mathbb N_n,\,j\in\mathbb N_h,h\ge 2,$ (or their equivalent Lagrange polynomials \cite{Mastroianni:2008:IPB:1502750})
using equally spaced interpolation points $s_{ij}\in \mathcal A_i,$
\begin{equation*}
s_{ij} = a_i+(j-1)\frac{b_i-a_i}{h-1},\quad j\in\mathbb N_h. 
\end{equation*}
The following result can be adapted from \cite{Mastroianni:2008:IPB:1502750}. 
\begin{thm}
\label{thm:high_1d}
Assume that the density function $t_{\mathfrak s}(\cdot|s)$ is $h$-times differentiable and define the constant 
\begin{align*}
&\mathcal M_h = \max_{s,\bar s\in\Upsilon} \left|\frac{\partial^h t_{\mathfrak s}(\bar s|s)}{\partial \bar s^h}\right|.
\end{align*}
The interpolation operator $\Pi_\Upsilon$, constructed with polynomial basis functions and equally spaced interpolation points, satisfies the inequality
\begin{equation*}
\|\Pi_{\Upsilon}\left(t_{\mathfrak s}(\cdot|s)\right)-t_{\mathfrak s}(\cdot|s)\|_{\infty} \le \mathcal E^{\mathfrak h} = \frac{\mathcal M_h}{4h}\left(\frac{\delta}{h-1}\right)^h,\quad \forall s\in\Upsilon,
\end{equation*}
where $\delta = \max_i \delta_i$, 
with $\delta_i = b_i-a_i,\,i\in\mathbb N_n,$
and where $h$ is the cardinality of the set of basis functions.
\end{thm}
Theorem~\ref{thm:high_1d} provides the necessary ingredients for Theorem~\ref{thm:error_dynamic}, 
leading to the quantification of the approximation error: 
employing Algorithm~\ref{algo:app_val} with equally spaced points and
polynomial basis functions of degree less than $h$,
the approximation error is upper bounded by the quantity
\begin{equation*}
\|\mu_t-\psi_t^{\mathfrak h}\|_\infty\le E_t^{\mathfrak h} = \mathcal E^{\mathfrak h}\kappa(t,M_{\mathfrak f}^{\mathfrak h}),\quad t\in\mathbb N_N,
\end{equation*}
with the constant $M_{\mathfrak f}^{\mathfrak h}$ defined in \eqref{eq:constant_M} and computed for this particular choice of basis functions and points.

It is worth highlighting that, 
unlike the piecewise constant case of Section~\ref{sec:partition}, 
with higher-order approximation approaches the global error is a nonlinear function of the partition size $\delta$, 
namely it depends on a power of the partition size contingent on the order of the selected interpolation operator. 
As such, its convergence speed, as $\delta$ is decreased, increases over that of the piecewise constant case. 

\medskip

\noindent\textbf{Example \ref{ex:linear_1d} (Continued).} 
We partition the set $\Upsilon = \cup_{i=1}^{n}\mathcal A_i$ for the one dimensional system of Example \ref{ex:linear_1d} with the intervals $\mathcal A_i = [a_i,b_i]$.
We select interpolation points $\{a_i,a_{i+1}\}$ with polynomial basis functions $\{1,s\}$, 
leading to piecewise affine approximations (namely, first-order interpolation with $h = 2$) of the density functions $\pi_t(\cdot)$. 
This set of basis functions can be equivalently transformed to 
\begin{equation*}
\Phi = \left\{\phi_{i1} = \frac{b_i-s}{b_i-a_i},\,\phi_{i2} = \frac{s-a_i}{b_i-a_i}\right\},
\end{equation*}
to obtain $\mathcal Q_i = \mathbb I_2$. 
The constant $M_{\mathfrak f}^{\mathfrak h}$ has the same value as $M_{\mathfrak f}=1/a$ and the quantity $\mathcal M_2$ in Theorem \ref{thm:high_1d} is $\mathcal M_2 = 1/\sigma^3\sqrt{2\pi}$. The error related to this first-order approximation can be upper bounded as
\begin{equation}
\label{eq:error_example_2}
\|\pi_t-\psi_t^{\mathfrak h}\|_\infty\le
\kappa(t,M_{\mathfrak f})\left[\frac{\delta^2}{\sigma^38\sqrt{2\pi}}+\frac{\phi_1(\alpha)}{\sigma}\right].
\end{equation}
Notice that the first part of the error in \eqref{eq:error_example_2}, 
which specifically relates to the first-order approximation,  
is proportional to $\delta^2$ -- this improves the error bound computed in \eqref{eq:error_example_1}.
Algorithm \ref{algo:app_val} is implemented for this linear system with the aforementioned parameters
$b=0,\beta_0 = 0,\gamma_0 = 1,\sigma = 0.1,$ and the time horizon $N=5$. 
The errors corresponding to the values $a = 1.2$ and $a = 0.8$ are analytically upper-bounded, 
for any $t\in\mathbb N_N$, as 
$\|\pi_t-\psi_t^{\mathfrak h}\|_\infty\le 179\delta^2+35.9\phi_1(\alpha)$
and as $\|\pi_t-\psi_t^{\mathfrak h}\|_\infty\le 409.3\delta^2+82.1\phi_1(\alpha)$,
respectively.  
The plots in Figure \ref{fig:higher_order_density} display the first- and zero-order approximations of the density function $\pi_N(\cdot)$ and compare it with the analytic solution for two different values $a = 1.2$ (left) and $a = 0.8$ (right). 
The partition size $n = 25$ and parameter $\alpha = 2.4$ have been selected in order to illustrate the differences of the two approximation methods in Figure \ref{fig:higher_order_density},
but may be increased at will in order to decrease the related error bound to match a desired value. 

\begin{figure}
\centering
\includegraphics[scale = 0.5]{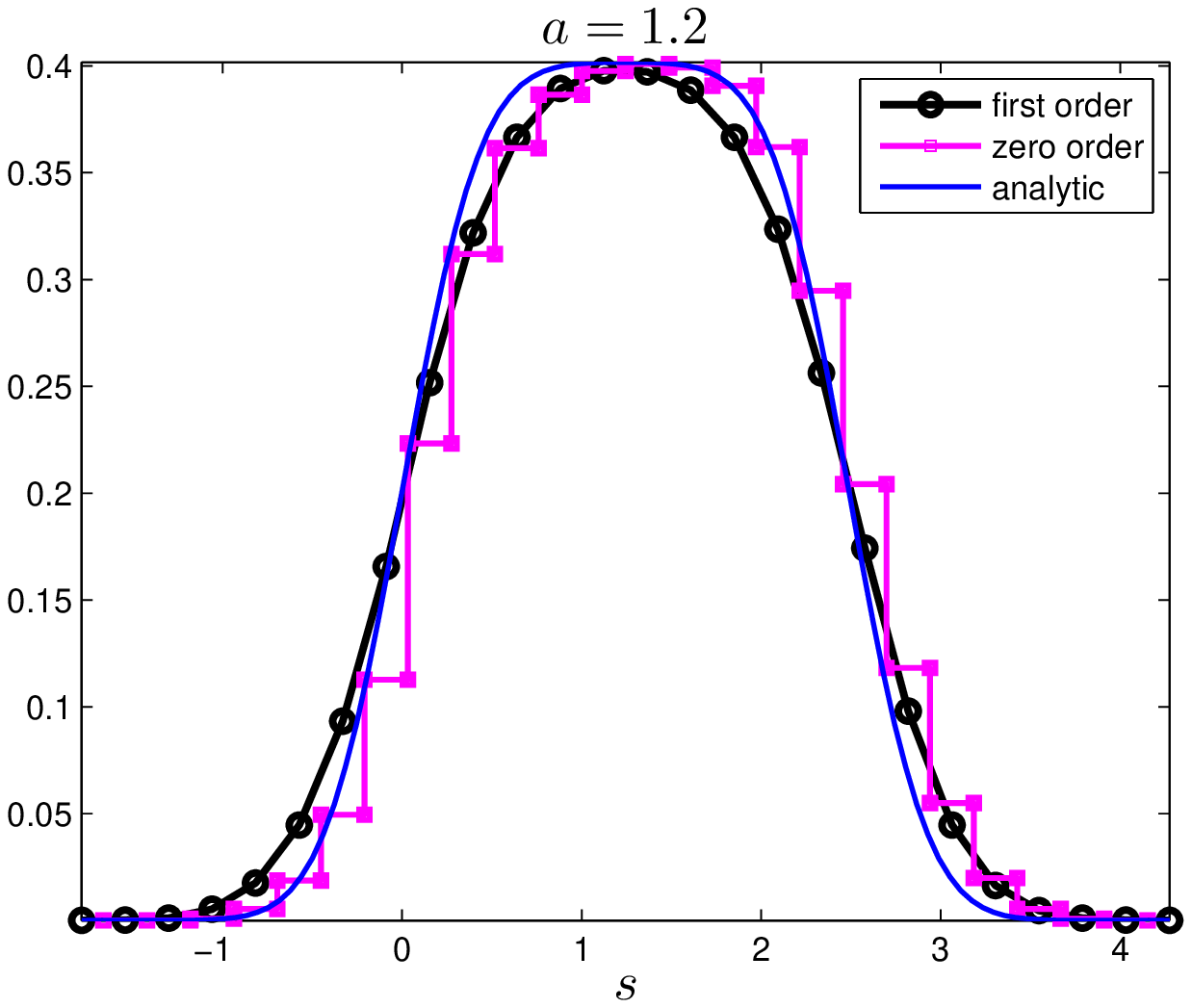}
\includegraphics[scale = 0.5]{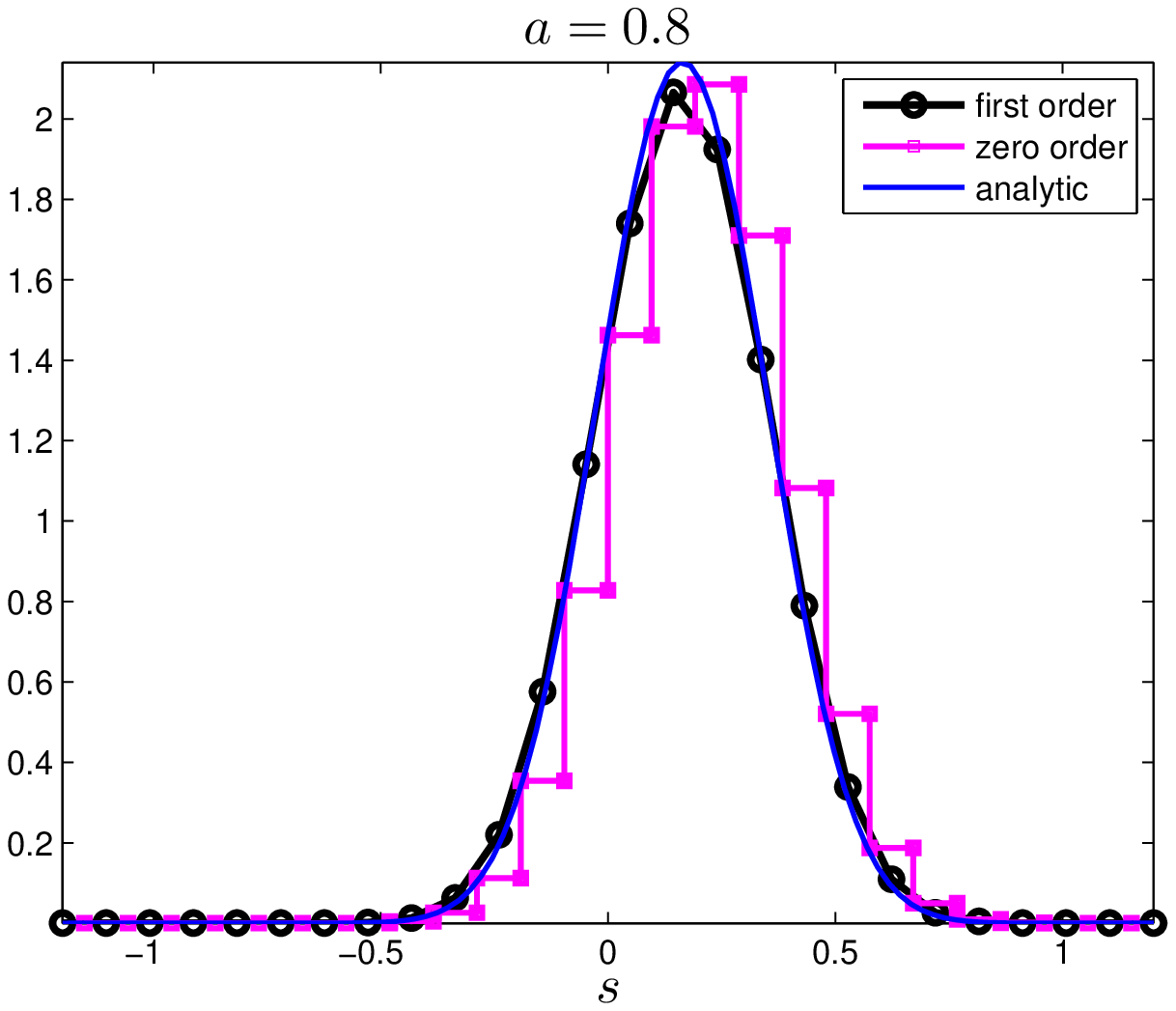}
\caption{Comparison of the first-order (affine) approximation $\psi_N^{\mathfrak h}(\cdot)$ 
versus the Markov chain abstraction (zero-order approximation, constant) $\psi_N(\cdot)$ of the state density function $\pi_N(\cdot)$ (derived analytically), 
for $N = 5$ and parameters $a=1.2$ (left) and $a=0.8$ (right).} 
\label{fig:higher_order_density}
\end{figure}

\subsection{Bilinear Interpolation for Two-Dimensional Systems}

We directly tailor the results of Section \ref{sec:approx&error} to a general two-dimensional Markov process, 
where $s = (s_1,s_2)\in\mathcal S = \mathbb R^2$. 
Assume that set $\Upsilon$ is replaced by a superset that is comprised of a finite union of rectangles: 
this replacement does not violate the bound on the truncation error formulated in Section \ref{sec:trunc}.   
Consider a uniform partition (using squared partition sets of size $\delta$) for the set $\Upsilon$.  
We employ a bilinear interpolation within each partition set
$\mathcal A_i = [a_{i1},b_{i1}]\times[a_{i2},b_{i2}]$
with basis
$\{\phi_{i1}(s) = 1,\phi_{i2}(s) = s_1,\phi_{i3}(s) = s_2,\phi_{i4}(s) = s_1 s_2\},i\in\mathbb N_n$
(or their equivalent Lagrange polynomials \cite{SAH12}).
Assume that the density function $t_{\mathfrak s}(\cdot|s)$ is partially differentiable and define the following bounds on its derivatives
\begin{align*}
\left|\frac{\partial^2 t_{\mathfrak s}(\bar s|s)}{\partial \bar s_k^2}\right|\le\mathcal M_2^k,\quad 
\left|\frac{\partial^3 t_{\mathfrak s}(\bar s|s)}{\partial \bar s_k^2\bar s_{3-k}}\right|\le\mathcal M_3^k, \quad
k\in\mathbb N_2,\, \forall s,\bar s\in\Upsilon.
\end{align*}
The operator $\Pi_\Upsilon$ in \eqref{eq:operator_partition}, 
constructed with bilinear interpolation within each partition set, 
satisfies the inequality
\begin{equation*}
\|\Pi_{\Upsilon}\left(t_{\mathfrak s}(\cdot|s)\right)-t_{\mathfrak s}(\cdot|s)\|_{\infty} \le \mathcal E^{\mathfrak h} = \frac{\delta^2}{16}\left(\mathcal M_2^1+\mathcal M_2^2\right)
+ \frac{\delta^3}{8\sqrt{2}}\left(\mathcal M_3^1+\mathcal M_3^2\right), 
\quad \forall s\in\Upsilon,
\end{equation*}
where $\delta = \max_i \delta_i$, 
with $\delta_i = \left[(b_{i1}-a_{i1})^2+(b_{i2}-a_{i2})^2\right]^{1/2},\,i\in\mathbb N_n.$
We implement Algorithm~\ref{algo:app_val} for two-dimensional processes using bilinear interpolation: 
the approximation error is upper-bounded by the quantity
\begin{equation*}
\|\mu_t-\psi_t^{\mathfrak h}\|_\infty\le E_t^{\mathfrak h} = \mathcal E^{\mathfrak h}\kappa(t,M_{\mathfrak f}^{\mathfrak h}), \quad t\in\mathbb N_N,
\end{equation*}
with the constant $M_{\mathfrak f}^{\mathfrak h}$ defined in \eqref{eq:constant_M} and computed for this particular choice of basis functions and points.
It can be proved that for bilinear interpolation basis functions,
the constant $M_{\mathfrak f}^{\mathfrak h}$ of \eqref{eq:constant_M} is upper bounded by $M_{\mathfrak f}$ of Assumption \ref{ass:Lip_cont_bounded}, and thus can be replaced by this quantity in the error computation.

\subsection{Trilinear Interpolation for Three-Dimensional Systems}

We now apply the results of Section \ref{sec:approx&error} to a general three-dimensional Markov process, 
where $s = (s_1,s_2,s_3)\in\mathcal S = \mathbb R^3.$ 
Again we replace the set $\Upsilon$ by a superset that is comprised of a finite union of boxes, 
without violating the bound on the truncation error formulated in Section \ref{sec:trunc}. 
Consider a uniform partition (using cubic sets of size $\delta$) for the set $\Upsilon$.  
We employ a trilinear interpolation within each partition set with basis functions
\begin{equation*}
\Phi = \{1,\,s_1,\,s_2,\,s_3,\,s_1s_2,\,s_2s_3,\,s_3s_1,\,s_1s_2s_3\}.
\end{equation*}
Assume that the density function $t_{\mathfrak s}(\cdot|s)$ is partially differentiable and define the following bounds on its derivatives
\begin{align*}
\left|\frac{\partial^2 t_{\mathfrak s}(\bar s|s)}{\partial \bar s_i^2}\right|\le\mathcal M^i_2,\quad
\left|\frac{\partial^3 t_{\mathfrak s}(\bar s|s)}{\partial \bar s_i^2\bar s_j}\right|\le\mathcal M^{ij}_3,\quad
\left|\frac{\partial^3 t_{\mathfrak s}(\bar s|s)}{\partial \bar s_1\bar s_2\bar s_3}\right|\le \mathcal M_3,\quad i,j\in\mathbb N_3,\,i\neq j.    
\end{align*}
We implement Algorithm~\ref{algo:app_val} for three-dimensional processes using trilinear interpolation in the operator $\Pi_\Upsilon$ \eqref{eq:operator_partition}.
The approximation error is then upper bounded by the quantity
\begin{equation*}
\|\mu_t-\psi_t^{\mathfrak h}\|_\infty\le \mathcal E^{\mathfrak h}\kappa(t,M_{\mathfrak f}^{\mathfrak h})\quad t\in\mathbb N_N,
\end{equation*}
with the constant
\begin{equation*}
\mathcal E^{\mathfrak h} = \frac{\delta^2}{24}\left(\mathcal M^1_2+\mathcal M^2_2+\mathcal M^3_2\right)
+\frac{\delta^3}{12\sqrt{3}}\left(\mathcal M^{12}_3+\mathcal M^{21}_3+\mathcal M^{23}_3+\mathcal M^{32}_3+\mathcal M^{13}_3+\mathcal M^{31}_3+3\mathcal M_3\right).
\end{equation*}
Similar to the bilinear interpolation case, the constant $M_{\mathfrak f}^{\mathfrak h}$
can be replaced by $M_{\mathfrak f}$ of Assumption \ref{ass:Lip_cont_bounded} in the error computation.

\section{Application of the Formal Approximation Procedure\\ to the Probabilistic Invariance Problem}
\label{sec:safety}
The problem of probabilistic invariance (or, equivalently, safety) for general Markov processes has been theoretically characterised in \cite{APLS08}  
and further investigated computationally in \cite{APKL10,SA11,SAH12,SA12}. 
With reference to a discrete-time Markov process $\mathscr M_{\mathfrak s}$ over a continuous state space $\mathcal S$, 
and to a safe set $\mathcal A \in \mathcal B(\mathcal S)$, 
the goal is to quantify the probability
\begin{equation*}
p_s^N(\mathcal A) = \mathbb P\{s(t)\in\mathcal A,\text{ for all } t\in\mathbb Z_N | s(0) = s\}. 
\end{equation*}
More generally, it is of interest to quantify the probability $p_{\pi_0}^N(\mathcal A)$, 
where the initial condition of the process $s(0)$ is a random variable characterised by the density function $\pi_0(\cdot)$.
In Section~\ref{subsec:safety_forward} we present a forward computation of probabilistic invariance by application of the approximation procedure above,  
then review results on backward computation \cite{APKL10,SA11,SAH12,SA12} in Section~\ref{subsec:safety_backward}. 
We conclude in Section~\ref{subsec:compare} with a comparison of the two approaches. 

\subsection{Forward Computation of Probabilistic Invariance}
\label{subsec:safety_forward}
The technique for approximating the density function of a process in time can be easily employed for the approximate computation of probabilistic invariance.
Define
functions $W_t:\mathcal S\rightarrow\mathbb R^{\ge 0}$, 
characterised as 
\begin{equation}
\label{eq:frwd_recursion}
W_{t+1}(\bar s) = \mathds 1_{\mathcal A}(\bar s)\int_{\mathcal S}W_t(s)t_{\mathfrak s}(\bar s|s)ds,\quad
W_0(\bar s) = \mathds 1_{\mathcal A}(\bar s)\pi_0(\bar s), \quad \forall \bar s\in\mathcal S.
\end{equation}
Then the solution of the problem is obtained as $p_{\pi_0}^N(\mathcal A) = \int_{\mathcal S}W_N(s)ds$.
A comparison of the recursions in \eqref{eq:frwd_recursion} and in \eqref{eq:rec_trunc} reveals how probabilistic invariance can be computed as a special case of the general approximation procedure in this work.
In applying the procedure, 
the only difference consists in replacing set $\Upsilon$ by the safe set $\mathcal A$, 
and in restricting Assumption~\ref{ass:Lip_cont_bounded} to hold over the safe set -- the solution over the complement of this set is trivially known, 
as such the error related to the truncation of the state space can be disregarded.  
The procedure consists in partitioning the safe set, 
in constructing the Markov chain $\mathscr M_{\mathfrak p}$ as per \eqref{eq:trans_elem}, 
and in computing $\psi_t(\cdot)$ as an approximation of $W_t(\cdot)$ based on \eqref{eq:approx_marg}. 
The error of this approximation is $\|W_t-\psi_t\|_\infty\le E_t$, which results in the following:
\begin{equation*}
\left|p_{\pi_0}^N(\mathcal A)-\int_{\mathcal A}\psi_t(s)ds \right|\le E_N\mathcal L(\mathcal A) = \kappa(N,M_{\mathfrak f})\lambda_{\mathfrak f}\delta\mathcal L(\mathcal A)\doteq E_{\mathfrak f}.
\end{equation*}
Note that the sub-density functions satisfy the inequalities
\begin{equation*}
1\ge \int_{\mathcal S}W_0(s)ds\ge \int_{\mathcal S}W_1(s)ds\ge \ldots\ge \int_{\mathcal S}W_N(s)ds\ge 0. 
\end{equation*}

\subsection{Backward Computation of Probabilistic Invariance}
\label{subsec:safety_backward}

The contributions in \cite{APKL10,SA11,SAH12,SA12} have characterised specifications in PCTL with an alternative formulation based on backward recursions.  
In particular, the computation of probabilistic invariance can be obtained via the value functions $V_t:\mathcal S\rightarrow[0,1]$, 
which are characterised as 
\begin{equation}
\label{eq:bckwd_recursion}
V_t(s) = \mathds 1_{\mathcal A}(s)\int_{\mathcal S}V_{t+1}(\bar s)t_{\mathfrak s}(\bar s|s)d\bar s,\quad
V_N(s) = \mathds 1_{\mathcal A}(s), \quad\forall s\in\mathcal S.
\end{equation}
The desired probabilistic invariance is expressed as 
$p_{\pi_0}^N(\mathcal A) = \int_{\mathcal S}V_0(s)\pi_0(s)ds.$
The value functions always map the state space to the interval $[0,1]$ and they are non-increasing, 
namely $V_{t}(s)\le V_{t+1}(s)$ for any fixed $s \in \mathcal S$.  
The contributions in \cite{APKL10,SA11,SAH12,SA12} discuss efficient algorithms for the approximate computation of the quantity $p_{\pi_0}^N(\mathcal A)$,   
relying on different assumptions on the model under study. 
The easiest and most straightforward procedure is based on the following assumption \cite{APKL10}. 
\begin{asm}
\label{ass:back_cond_kernel}
The conditional density function of the process is globally Lipschitz continuous with respect to the conditional state within the safe set.   
Namely, there exists a finite constant $\lambda_{\mathfrak b}$, such that
\begin{equation*}
|t(\bar s|s)-t(\bar s|s')|\le \lambda_{\mathfrak b} \|s-s'\|, \quad \forall s,s',\bar s\in\mathcal A.
\end{equation*}
A finite constant $M_{\mathfrak b}$ is introduced as 
$M_{\mathfrak b} = \sup_{s\in\mathcal A}\int_{\mathcal A}t_{\mathfrak s}(\bar s|s)d\bar s\le 1$.
\end{asm}

The procedure introduces a partition of the safe set $\mathcal A = \cup_{i=1}^{n}\mathcal A_i$ and extends it to $\mathcal S = \cup_{i=1}^{n+1}\mathcal A_i$, 
with $\mathcal A_{n+1} = \mathcal S\backslash\mathcal A$.
Then it selects arbitrary representative points $s_i\in \mathcal A_i$
and constructs a finite-state Markov chain $\mathscr M_{\mathfrak b}$ over the finite state space $\{s_1,s_2,\ldots,s_{n+1}\}$, 
endowed with transition probabilities
\begin{equation}
\label{eq:trans_elem_rep}
\begin{array}{l}
P(s_i,s_j) = \int_{\mathcal A_j}t_{\mathfrak s}(\bar s|s_i)d\bar s,\quad P(s_{n+1},s_j) = \delta_{(n+1)j}, 
\end{array}
\end{equation}
for all $i\in\mathbb N_n, j\in\mathbb N_{n+1}$. The error of such an approximation is \cite{SAH12}:
\begin{equation*}
E_{\mathfrak b} \doteq \kappa(N,M_{\mathfrak b})\lambda_{\mathfrak b}\delta \mathcal L(\mathcal A),
\end{equation*} 
where $\delta$ is the max partitions diameter, 
and $\mathcal L(\mathcal A)$ is the Lebesgue measure of set $\mathcal A$. 

\subsection{Comparison of the Two Approaches}
\label{subsec:compare}

We first compare the two constructed Markov chains.
The Markov chain $\mathscr M_{\mathfrak p}$ obtained with the abstraction from the forward approach is a special case of the Markov chain $\mathscr M_{\mathfrak b}$ from the backward approach: in the latter case in fact the representative points can be selected intelligently to determine the average probability of jumping from one partition set to another. 
More specifically, the quantities \eqref{eq:trans_elem} are a special case of those in \eqref{eq:trans_elem_rep} (based on the mean value theorem for integration).
We will show that this leads to a less conservative (smaller) error bound for the approximation. 

The forward computation is in general more informative than the backward computation since it provides not only the solution of the safety problem in time, but also the state distribution over the safe set.
Further the forward approach may provide some insight to the solution of the infinite-horizon safety problem \cite{ta2011,TA14} for a given initial distribution. 
As discussed in \cite{TA14},
solution of the infinite-horizon safety problem depends on the existence of absorbing subsets of the safe set. 
The outcome of the forward approach can provide evidence on the non existence of such subsets.
Finally, the forward approach presented in Sections \ref{sec:trunc}-\ref{sec:proj} for approximating density functions can be used to approximate the value functions in the recursion \eqref{eq:frwd_recursion} over \emph{unbounded} safe sets since we do not require the state space (thus also the safe set) to be bounded, 
while boundedness of the safe set is required in all the results in the literature that are based on backward computations. 

Next, we compare errors and related assumptions. 
The error computations
rely on two different assumptions: 
the Lipschitz continuity of the conditional density function with respect to the current state or to the next state, respectively.  
Further, 
the constants $M_{\mathfrak f}$ and $M_{\mathfrak b}$ are generally different and play an important role in the form of the error. 
$M_{\mathfrak b}$ represents the maximum probability of remaining within a given set, 
while $M_{\mathfrak f}$ is an indication of the maximum concentration of the process evolution towards one state, over a single time-step. 
$M_{\mathfrak b}$ is always less than or equal to one, while $M_{\mathfrak f}$ could be any finite positive number.

\medskip

\noindent\textbf{\textbf{Example \ref{ex:linear_1d} (Continued).}}
The constants $\lambda_{\mathfrak f},M_{\mathfrak f}$ and $\lambda_{\mathfrak b},M_{\mathfrak b}$ for the one dimensional dynamical system of Example \ref{ex:linear_1d} are
\begin{equation*}
\quad \lambda_{\mathfrak f} = \frac{1}{\sigma^2\sqrt{2\pi e}},\quad \lambda_{\mathfrak b} = a\lambda_{\mathfrak f}, \quad M_{\mathfrak f} =\frac{1}{a}, \quad M_{\mathfrak b}\le 1.
\end{equation*}
If $0<a<1$, 
the system trajectories converge to an equilibrium point (in expected value). 
In this case the model solution has higher chances of ending up in a neighbourhood of the equilibrium in time, 
and the backward recursion provides a better error bound. 
If $a>1$, 
the system trajectories tend to diverge with time. 
In this case the forward recursion provides a much better error bound, 
compared to the backward recursion. 

For the numerical simulation we select a safety set $\mathcal A = [0,1]$, 
a noise level $\sigma = 0.1$, 
and a time horizon $N = 10$. 
The solution of the safety problem for the two cases $a = 1.2$ and $a=0.8$ is plotted in Figure~\ref{fig:1d_example}.
We have computed constants $\lambda_{\mathfrak f} = 24.20, M_{\mathfrak b} = 1$ (in both cases), 
while $\lambda_{\mathfrak b} = 29.03, M_{\mathfrak f} = 0.83$ for the first case  
and $\lambda_{\mathfrak b} = 19.36, M_{\mathfrak f} = 1.25$ for the second case. 
We have selected the center of the partition sets (distributed uniformly over the set $\mathcal A$) as representative points for the Markov chain $\mathscr M_{\mathfrak b}$.
In order to compare the two approaches, we have assumed the same computational effort (related to the same partition size of $\delta = 0.7\times 10^{-4}$),
and have obtained an error $E_{\mathfrak f} = 0.008,E_{\mathfrak b} = 0.020$ for $a=1.2$ and $E_{\mathfrak f} = 0.056, E_{\mathfrak b} = 0.014$ for $a=0.8$.
The simulations show that the forward approach works better for $a=1.2$, 
while the backward approach is better suitable for $a=0.8$.
Note that the approximate solutions provided by the two approaches are very close: 
the difference of the transition probabilities computed via the Markov chains $\mathcal M_{\mathfrak f},\mathcal M_{\mathfrak b}$ are in the order of $10^{-8}$, 
and the difference in the approximate solutions (black curve in Figure~\ref{fig:1d_example}) is in the order of $10^{-6}$. 
This has been due to the selection of very fine partition sets that have resulted in small abstraction errors. \hfill\qed
\begin{figure}
\centering
\includegraphics[scale = 0.5]{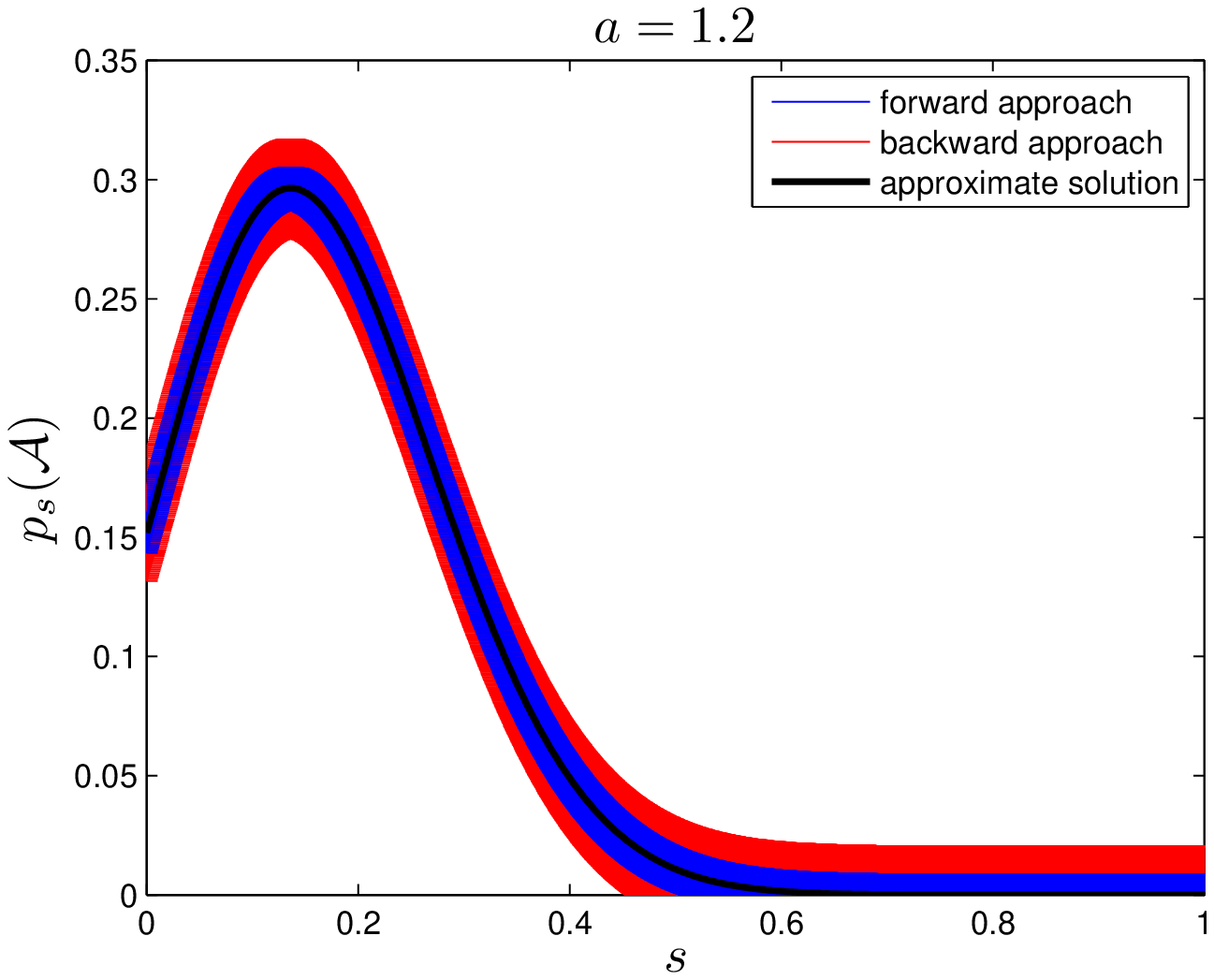}
\includegraphics[scale = 0.5]{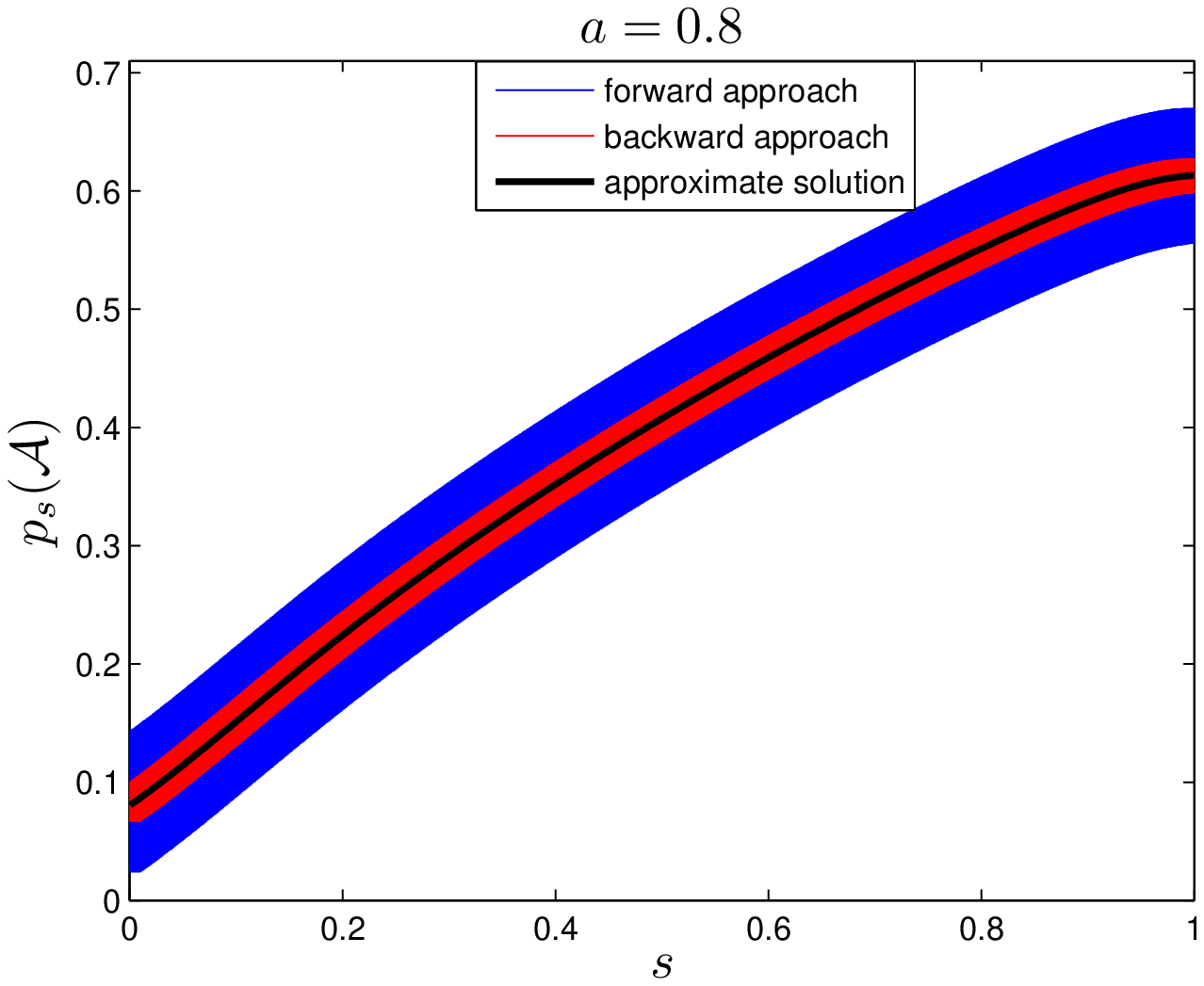}
\caption{Approximate solution of the probabilistic invariance problem (thin black line), 
together with error intervals of forward (\textcolor{blue}{blue} band) and backward (\textcolor{red}{red} band) approaches, for $a=1.2$ (left) and $a=0.8$ (right).}
\label{fig:1d_example}
\end{figure}

\begin{rem}
Over deterministic models,
\cite{IM07HSCC} compares forward and backward reachability analysis and provides insights on their differences:
the claim is that for systems with significant contraction, forward reachability is more effective than backward reachability because of numerical stability issues.
On the other hand, 
for the probabilistic models under study, 
the result indicates that under Lipschitz continuity of the transition kernel the backward approach is more effective in systems with convergence in the state distribution. 
If we treat deterministic systems as special (limiting) instances of stochastic systems, 
our result is not contradicting with \cite{IM07HSCC} since the Lipschitz continuity assumption on the transition kernels of probabilistic models does not hold over deterministic ones. \hfill\qed
\end{rem}
Motivated by the previous example, 
we study how the convergence properties of a Markov process are related to the constant $M_{\mathfrak f}$.
 
\begin{thm}
\label{thm:exp_conv}
Assume that the initial density function $\pi_{0}(s)$ is bounded and that
the constant $M_{\mathfrak f}$ is finite and $M_{\mathfrak f}<1$. 
If the state space is unbounded, the sequence of density functions $\{\pi_t(s)|t\ge 0\}$ uniformly exponentially converges to zero. 
The sequence of probabilities $\mathbb P\{s(t)\in \mathcal A \}$ and the corresponding solution of the safety problem for any compact safe set $\mathcal A$ exponentially converge to zero. 
\end{thm}
Theorem~\ref{thm:exp_conv} indicates that under the invoked assumptions the probability ``spreads out'' over the unbounded state space as time progresses.  
Moreover, the theorem ensures the absence of absorbing sets \cite{ta2011,TA14}, 
which are indeed known to characterise the solution of infinite-horizon properties.
Example~\ref{ex:linear_stab} studies the relationship between constant $M_{\mathfrak f}$ and the stability of linear stochastic difference equations.
\begin{exa}
\label{ex:linear_stab}
Consider the stochastic linear difference equations
\begin{equation*}
s(t+1) = A s(t) + w(t),\quad s(\cdot),w(\cdot)\in\mathbb R^d,
\end{equation*}
where $w(\cdot)$ are i.i.d. random vectors with known distributions.
For such systems $M_{\mathfrak f} = 1/|\det A|$, then the condition $M_{\mathfrak f}<1$ implies instability of the system in expected value.
Equivalently, mean-stability of the system implies $M_{\mathfrak f}\ge 1$. 
Note that for this class of systems $M_{\mathfrak f}>1$ does not generally imply stability, since $\det A$ is only the product of the eigenvalues of the system. \hfill\qed
\end{exa}
The Lipschitz constants $\lambda_{\mathfrak f}$ and $\lambda_{\mathfrak b}$ have a different nature, as clarified in Example~\ref{ex:non_linear}. 
\begin{exa}
\label{ex:non_linear}
Consider the dynamical system 
\begin{equation*}
s(t+1) = f(s(t),w(t)),\quad s(\cdot),w(\cdot)\in\mathbb R^d,
\end{equation*}
where $w(\cdot)$ are i.i.d. with known distribution $t_w(\cdot)$. 
Suppose that the vector field $f:\mathbb R^d\times\mathbb R^d\rightarrow\mathbb R^d$ is continuously differentiable and that the matrix $\frac{\partial f}{\partial w}$ is invertible. 
Then the \emph{implicit function theorem} guarantees the existence and uniqueness of a function $g:\mathbb R^d\times\mathbb R^d\rightarrow\mathbb R^d$ such that $w(t) = g(s(t+1),s(t))$. 
The conditional density function of the system in this case is \cite{Papoulis91}:
\begin{equation*}
t_{\mathfrak s}(\bar s|s) = \left|\det\left[\frac{\partial g}{\partial \bar s}(\bar s,s)\right]\right|t_w(g(\bar s,s)).
\end{equation*}
The Lipschitz constants $\lambda_{\mathfrak f},\lambda_{\mathfrak b}$ are specified by the dependence of function $g(\bar s,s)$ from the variables $\bar s,s$, respectively.
As a special case the invertibility of $\frac{\partial f}{\partial w}$ is guaranteed for systems with additive process noise, namely
$f(s,w) = f_a(s)+w$. Then $g(\bar s,s) = \bar s-f_a(s)$,
$\lambda_{\mathfrak f}$ is the Lipschitz constant of $t_w(\cdot)$,
while $\lambda_{\mathfrak b}$ is the multiplication of the Lipschitz constant of $t_w(\cdot)$ and of $f_a(\cdot)$. \hfill\qed
\end{exa}

\section{Conclusions} 
\label{sec:concl}

This contribution has put forward new algorithms, 
based on Markov chain abstractions, 
for the efficient computation of approximate solutions of the state distribution function in time of Markov processes evolving over continuous state spaces. 
A higher-order function approximation method has also been presented, with a formal derivation of an upper bound on the associated error.  
The approach has been applied to the verification of a particular non-nested PCTL formula (expressing probabilistic safety or invariance), 
and compared with an alternative computational approach from the literature.

The authors plan to integrate the presented procedures within the software tool \textsf{FAUST}$^{\mathsf 2}$, 
which is developed for the formal abstraction and verification of uncountable-state stochastic processes \cite{FAUST15}. 
The software enables the user to automatically export the finite-state abstracted model to existing probabilistic model checking tools,  
such as PRISM and MRMC \cite{HKNP06,KKZ05}, for further quantitative analysis, verification, and synthesis objectives.

\bibliographystyle{plain}
\bibliography{biblio}

\begin{thebibliography}{10}

\bibitem{A12}
A.~Abate.
\newblock Approximation metrics based on probabilistic bisimulations for
  general state-space markov processes: a survey.
\newblock {\em Electron. Notes Theor. Comput. Sci.}, 2012.

\bibitem{APKL10}
A.~Abate, J.-P. Katoen, J.~Lygeros, and M.~Prandini.
\newblock Approximate model checking of stochastic hybrid systems.
\newblock {\em European Journal of Control}, 6:624--641, 2010.

\bibitem{AKNP14}
A.~Abate, M.~Kwiatkowska, G.~Norman, and D.~Parker.
\newblock Probabilistic model checking of labelled markov processes via finite
  approximate bisimulations.
\newblock In F.~van Breugel, E.~Kashefi, C.~Palamidessi, and J.~Rutten,
  editors, {\em Horizons of the Mind -- P. Panangaden Festschrift}, Lecture
  Notes in Computer Science 8464, pages 40--58. Springer Verlag, 2014.

\bibitem{APLS08}
A.~Abate, M.~Prandini, J.~Lygeros, and S.~Sastry.
\newblock Probabilistic reachability and safety for controlled discrete time
  stochastic hybrid systems.
\newblock {\em Automatica}, 44(11):2724--2734, November 2008.

\bibitem{BKH99}
C.~Baier, J.-P. Katoen, and H.~Hermanns.
\newblock Approximate symbolic model checking of continous-time {Markov}
  chains.
\newblock In J.C.M. Baeten and S.~Mauw, editors, {\em Concurrency Theory},
  volume 1664 of {\em Lecture Notes in Computer Science}, pages 146--162.
  Springer Verlag, Berlin Heidelberg, 1999.

\bibitem{CDPP14}
P.~Chaput, V.~Danos, P.~Panangaden, and G.~Plotkin.
\newblock Approximating {M}arkov processes by averaging.
\newblock {\em J. ACM}, 61(1):5:1--5:45, January 2014.

\bibitem{DDP03}
V.~Danos, J.~Desharnais, and P.~Panangaden.
\newblock Conditional expectation and the approximation of labelled {M}arkov
  processes.
\newblock In Roberto Amadio and Denis Lugiez, editors, {\em CONCUR 2003 -
  Concurrency Theory}, volume 2761 of {\em Lecture Notes in Computer Science},
  pages 477--491. Springer Berlin Heidelberg, 2003.

\bibitem{DDP04}
V.~Danos, J.~Desharnais, and P.~Panangaden.
\newblock Labelled {M}arkov processes: Stronger and faster approximations.
\newblock {\em Electr. Notes Theor. Comput. Sci.}, 87:157--203, 2004.

\bibitem{DGJP03}
J.~Desharnais, V.~Gupta, R.~Jagadeesan, and P.~Panangaden.
\newblock Approximating labelled {M}arkov processes.
\newblock {\em Information and Computation}, 184(1):160 -- 200, 2003.

\bibitem{SA11}
S.~{Esmaeil Zadeh Soudjani} and A.~Abate.
\newblock Adaptive gridding for abstraction and verification of stochastic
  hybrid systems.
\newblock In {\em Proceedings of the 8th International Conference on
  Quantitative Evaluation of Systems}, pages 59--69, September 2011.

\bibitem{SAH12}
S.~{Esmaeil Zadeh Soudjani} and A.~Abate.
\newblock {H}igher-{O}rder {A}pproximations for {V}erification of {S}tochastic
  {H}ybrid {S}ystems.
\newblock In S.~Chakraborty and M.~Mukund, editors, {\em Automated Technology
  for Verification and Analysis}, volume 7561 of {\em Lecture Notes in Computer
  Science}, pages 416--434. Springer Verlag, Berlin Heidelberg, 2012.

\bibitem{SA12}
S.~{Esmaeil Zadeh Soudjani} and A.~Abate.
\newblock Probabilistic invariance of mixed deterministic-stochastic dynamical
  systems.
\newblock In {\em {ACM} Proceedings of the 15th {I}nternational {C}onference on
  {H}ybrid {S}ystems: {C}omputation and {C}ontrol}, pages 207--216, Beijing,
  PRC, April 2012.

\bibitem{SA13}
S.~{Esmaeil Zadeh Soudjani} and A.~Abate.
\newblock Adaptive and sequential gridding procedures for the abstraction and
  verification of stochastic processes.
\newblock {\em SIAM Journal on Applied Dynamical Systems}, 12(2):921--956,
  2013.

\bibitem{SA14}
S.~{Esmaeil Zadeh Soudjani} and A.~Abate.
\newblock Precise approximations of the probability distribution of a {M}arkov
  process in time: an application to probabilistic invariance.
\newblock In E.~Abraham and K.~Havelund, editors, {\em International Conference
  on Tools and Algorithms for the Construction and Analysis of Systems
  (TACAS)}, volume 8413 of {\em Lecture Notes in Computer Science}, pages
  547--561. Springer Verlag, 2014.

\bibitem{SATAC12}
S.~{Esmaeil Zadeh Soudjani} and A.~Abate.
\newblock Probabilistic reach-avoid computation for partially-degenerate
  stochastic processes.
\newblock {\em IEEE Transactions on Automatic Control}, 59(2):528--534, 2014.

\bibitem{FAUST15}
S.~{Esmaeil Zadeh Soudjani}, C.~Gevaerts, and A.~Abate.
\newblock \textsf{FAUST}$^{\textsf{2}}$: Formal abstractions of
  uncountable-state stochastic processes.
\newblock In C.~Baier and C.~Tinelli, editors, {\em Tools and Algorithms for
  the Construction and Analysis of Systems (TACAS)}, volume 9035 of {\em
  Lecture Notes in Computer Science}, pages 272--286. Springer Verlag, 2015.

\bibitem{HKNP06}
A.~Hinton, M.~Kwiatkowska, G.~Norman, and D.~Parker.
\newblock {PRISM}: A tool for automatic verification of probabilistic systems.
\newblock In H.~Hermanns and J.~Palsberg, editors, {\em Tools and Algorithms
  for the Construction and Analysis of Systems}, volume 3920 of {\em Lecture
  Notes in Computer Science}, pages 441--444. Springer Verlag, Berlin
  Heidelberg, 2006.

\bibitem{KKZ05}
J.-P. Katoen, M.~Khattri, and I.~S. Zapreev.
\newblock A {M}arkov reward model checker.
\newblock In {\em IEEE Proceedings of the International Conference on
  Quantitative Evaluation of Systems}, pages 243--244, Los Alamos, CA, USA,
  2005.

\bibitem{KR06}
K.~Koutsoukos and D.~Riley.
\newblock Computational methods for reachability analysis of stochastic hybrid
  systems.
\newblock In J.~Hespanha and A.~Tiwari, editors, {\em Hybrid Systems:
  Computation and Control}, volume 3927 of {\em Lecture Notes in Computer
  Science}, pages 377--391. Springer Verlag, Berlin Heidelberg, 2006.

\bibitem{Kur94}
R.~P. Kurshan.
\newblock {\em Computer-Aided Verification of Coordinating Processes: The
  Automata-Theoretic Approach}.
\newblock Princeton Series in Computer Science. Princeton University Press,
  1994.

\bibitem{KD01}
H.~J. {Kushner} and P.G. Dupuis.
\newblock {\em Numerical Methods for Stochastic Control Problems in Continuous
  Time}.
\newblock Springer-Verlag, New York, 2001.

\bibitem{mpt}
M.~Kvasnica, P.~Grieder, and M.~Baoti\'{c}.
\newblock Multi-parametric toolbox {(MPT)}, 2004.

\bibitem{KNSS00}
M.~Kwiatkowska, G.~Norman, R.~Segala, and J.~Sproston.
\newblock Verifying quantitative properties of continuous probabilistic timed
  automata.
\newblock In {\em {CONCUR'00}}, volume 1877 of {\em Lecture Notes in Computer
  Science}, pages 123--137. Springer Verlag, Berlin Heidelberg, 2000.

\bibitem{Mastroianni:2008:IPB:1502750}
G.~Mastroianni and G.V. Milovanovic.
\newblock {\em Interpolation Processes: Basic Theory and Applications}.
\newblock Springer Verlag, 2008.

\bibitem{IM07HSCC}
I.M. Mitchell.
\newblock Comparing forward and backward reachability as tools for safety
  analysis.
\newblock In {\em Proceedings of the 10th international conference on Hybrid
  systems: computation and control}, HSCC'07, pages 428--443, Berlin,
  Heidelberg, 2007. Springer-Verlag.

\bibitem{Papoulis91}
A.~Papoulis.
\newblock {\em Probability, Random Variables, and Stochastic Processes}.
\newblock Mc{G}raw-Hill, 3rd edition, 1991.

\bibitem{PH06}
M.~Prandini and J.~Hu.
\newblock Stochastic reachability: Theory and numerical approximation.
\newblock In C.G. Cassandras and J.~Lygeros, editors, {\em Stochastic hybrid
  systems}, Automation and Control Engineering Series 24, pages 107--138.
  Taylor \& Francis Group/CRC Press, 2006.

\bibitem{ta2011}
I.~Tkachev and A.~Abate.
\newblock On infinite-horizon probabilistic properties and stochastic
  bisimulation functions.
\newblock In {\em Proceedings of the 50th IEEE Conference on Decision and
  Control and European Control Conference}, pages 526--531, Orlando, FL,
  December 2011.

\bibitem{TA13}
I.~Tkachev and A.~Abate.
\newblock Formula-free {F}inite {A}bstractions for {L}inear {T}emporal
  {V}erification of {S}tochastic {H}ybrid {S}ystems.
\newblock In {\em Proceedings of the 16th International Conference on Hybrid
  Systems: Computation and Control}, pages 283--292, Philadelphia, PA, April
  2013.

\bibitem{TA14}
I.~Tkachev and A.~Abate.
\newblock Characterization and computation of infinite-horizon specifications
  over {M}arkov processes.
\newblock {\em Theoretical Computer Science}, 515(0):1--18, 2014.

\end{thebibliography}

\vspace{2cm}
\appendix
\section{Proof of Statements}
\begin{proof}[Proof of Theorem \ref{thm:bound_behaviour}]
We prove the theorem inductively.
The statement is trivial for $t=0$ based on Assumption \ref{ass:marg_conv}.
Suppose it is true for $t$, then we prove it for $t+1$.
Take $\bar s\in\mathcal S\backslash\Lambda_{t+1}$, then
\begin{align*}
\pi_{t+1}(\bar s) = \int_{\mathcal S}t_{\mathfrak s}(\bar s|s)\pi_t(s)ds
& = \int_{\mathcal S\backslash\Lambda_t}t_{\mathfrak s}(\bar s|s)\pi_t(s)ds + \int_{\Lambda_t}t_{\mathfrak s}(\bar s|s)\pi_t(s)ds.
\end{align*}
The first integral is upper bounded by $\varepsilon_t M_{\mathfrak f}$.
The domain of the second integral implies that $(s,\bar s)\in\Lambda_t\times\mathcal S$. Combining this with the definition \eqref{eq:recur_supp} of $\Lambda_{t+1}$ and $\bar s\notin\Lambda_{t+1}$ results in $(s,\bar s)\notin \varGamma$. Then 
\begin{align*}
\int_{\Lambda_t}t_{\mathfrak s}(\bar s|s)\pi_t(s)ds\le \epsilon \int_{\Lambda_t}\pi_t(s)ds\le \epsilon.
\end{align*}
Then we obtain $\pi_{t+1}(\bar s)\le \varepsilon_t M_{\mathfrak f}+\epsilon=\varepsilon_{t+1}$, 
for all $\bar s\in\mathcal S\backslash\Lambda_{t+1}$.
\end{proof}

\proof[Proof of Theorem \ref{thm:error_trunc}]
The initial density function $\mu_0$ satisfies the following inequality:
\begin{align*}
\|\pi_0-\mu_0\|_{\infty} = \|\pi_0-\mathds 1_{\Lambda_0}\pi_0\|_{\infty}
= \|\mathds 1_{\mathcal S\backslash\Lambda_0}\pi_0\|_{\infty}
= \sup\left\{\pi_0(s),\, s\in\mathcal S\backslash\Lambda_0\right\}\le \varepsilon_0.
\end{align*}
Suppose $\mu_t$ satisfies the inequality. We prove that it is also true for $\mu_{t+1}$.
Take any $\bar s\in\mathcal S\backslash\Upsilon$,
\begin{equation*}
\mathcal S\backslash\Upsilon\subset \mathcal S\backslash\Lambda_{t+1}\Rightarrow|\pi_{t+1}(\bar s)-\mu_{t+1}(\bar s)| = \pi_{t+1}(\bar s)\le \varepsilon_{t+1}.
\end{equation*}
Take any $\bar s\in\Upsilon$, we have
\begin{equation*}
|\pi_{t+1}(\bar s)-\mu_{t+1}(\bar s)| \le \int_{\mathcal S}t_{\mathfrak s}(\bar s|s)|\pi_t(s)-\mu_t(s)|ds
\le \varepsilon_t\int_{\mathcal S}t_{\mathfrak s}(\bar s|s)ds\le \varepsilon_tM_{\mathfrak f}\le\varepsilon_{t+1}.\eqno{\qEd}
\end{equation*}

\proof[Proof of Lemma \ref{lmm:lip_cont}] 
For any $t\in\mathbb N$ and $\bar s,\bar s'\in\mathcal S$,
\begin{equation*}
|\pi_t(\bar s)-\pi_t(\bar s')|
\le \int_{\mathcal S}\pi_{t-1}(s)\left|t_{\mathfrak s}(\bar s|s)-t_{\mathfrak s}(\bar s'|s)\right|ds
\le  \lambda_{\mathfrak f}\|\bar s-\bar s'\|\int_{\mathcal S}\pi_t(s)ds =  \lambda_{\mathfrak f}\|\bar s-\bar s'\|.\eqno{\qEd}
\end{equation*}

\begin{proof}[Proof of Theorem \ref{thm:Error}]
We use the triangle inequality as 
\begin{equation*}
\|\pi_t-\psi_t\|_\infty \le \|\pi_t-\mu_t\|_\infty + \|\mu_t-\psi_t\|_\infty
\le \varepsilon_t + \|\mu_t-\psi_t\|_\infty.
\end{equation*}
Define the set of Lebesgue integrable functions $\mathcal D$ and its subset $\mathcal D_1$:
\begin{align*}
& \mathcal D = \left\{f:\mathcal S\rightarrow\mathbb R^{\ge 0}\bigg|\int_{\mathcal S}f(s)ds<\infty\right\},\\
& \mathcal D_1 = \left\{f\in\mathcal D\bigg|\int_{\mathcal S}t_{\mathfrak s}(\bar s|s)f(s)ds<\infty,\text{ for all }\bar s\in\mathcal S\right\},
\end{align*}
and the operators
\begin{align*}
& \mathcal R:\mathcal D_1\rightarrow\mathcal D,\quad
\mathcal R(f)(\bar s) =  \mathds 1_{\Upsilon}(\bar s)\int_{\mathcal S}t_{\mathfrak s}(\bar s|s)f(s)ds,\\
& \mathcal R_a:\mathcal D\rightarrow\mathcal D_1,\quad
\mathcal R_a(f)(\bar s) =  \sum_{i=1}^{n}\frac{\int_{\mathcal A_i}f(s)ds}{\mathcal L(\mathcal A_i)}\mathds 1_{\mathcal A_i}(\bar s).
\end{align*}
Then $\mu_t(s)$ is formulated as $t-$times application of the operator $\mathcal R$ to the initial function $\pi_0(s)$,
\begin{align*}
\mu_{t+1} = \mathcal R(\mu_t)\Rightarrow\mu_t = \mathcal R^t(\pi_0) \quad\forall t\in\mathbb N.
\end{align*}
Define functions $\psi_t$ by the recursive equation:
\begin{align*}
\psi_{t+1} =(\mathcal R_a\mathcal R )(\psi_t)\Rightarrow\psi_t = (\mathcal R_a\mathcal R )^t(\psi_0) \quad\forall t\in\mathbb N,
\end{align*}
initialised with $\psi_0 = \mathcal R_a(\pi_0)$.
Notice that the functions $\psi_t,\,t\in\mathbb N,$ are all piecewise constant due to their recursive definition. We restrict our attention to the set $\Upsilon$ since the supports of both functions $\mu_t,\psi_t$ are included in $\Upsilon$. 
The goal is to
prove that $\psi_t$ satisfies \eqref{eq:approx_marg} and $\|\mu_t-\psi_t\|_\infty\le E_t$.
We achieve this goal using induction. The initial function $\psi_0$ is of the form
\begin{equation*}
\psi_{0}(s) = \mathcal R_a(\pi_0)(s)
= \sum_{i=1}^{n}\frac{\int_{\mathcal A_i}\pi_0(v)dv}{\mathcal L(\mathcal A_i)}\mathds 1_{\mathcal A_i}(s)
= \sum_{i=1}^{n}\frac{p_0(i)}{\mathcal L(\mathcal A_i)}\mathds 1_{\mathcal A_i}(s).
\end{equation*}
Suppose the statement is true for $t$. Then for any $s\in\mathcal A_j$, $j\in\mathbb N_n$,
\begin{align*}
\psi_{t+1}(s) & =(\mathcal R_a\mathcal R )(\psi_t)(s)
= \frac{1}{\mathcal L(\mathcal A_j)}\int_{\mathcal A_j}\mathcal R(\psi_t)(v)dv\\
& = \frac{1}{\mathcal L(\mathcal A_j)}\int_{\mathcal A_j}\int_{\mathcal S}t_{\mathfrak s}(v|u)\psi_t(u)du dv
= \frac{1}{\mathcal L(\mathcal A_j)}\sum_{i=1}^{n}\int_{\mathcal A_j}\int_{\mathcal A_i}t_{\mathfrak s}(v|u)\psi_t(u)du dv\\
& = \frac{1}{\mathcal L(\mathcal A_j)}\sum_{i=1}^{n}
\frac{p_t(i)}{\mathcal L(\mathcal A_i)}\int_{\mathcal A_j}\int_{\mathcal A_i}t_{\mathfrak s}(v|u)du dv
= \frac{1}{\mathcal L(\mathcal A_j)}\sum_{i=1}^{n}p_t(i) P_{ij}\\
&= \frac{1}{\mathcal L(\mathcal A_j)}p_{t+1}(j).
\end{align*}
We have proved \eqref{eq:approx_marg}.
The approximation error for $t=0$ is computed using \eqref{eq:cond_init_dens} in Assumption \ref{ass:Lip_cont_bounded}. Take any state $s\in\mathcal A_i$, $i\in\mathbb N_n$:
\begin{equation*}
|\mu_0(s)-\psi_0(s)| = \left|\pi_0(s)- \frac{\int_{\mathcal A_i}\pi_0(v)dv}{\mathcal L(\mathcal A_i)}\right|
\le \frac{\int_{\mathcal A_i}|\pi_0(s)-\pi_0(v)|dv}{\mathcal L(\mathcal A_i)}\le \lambda_0\delta = E_0.
\end{equation*}
For $t\ge 1$ we can write
\begin{equation*}
\|\mu_{t+1}-\psi_{t+1}\|_\infty = \|\mathcal R(\mu_t)-(\mathcal R_a\mathcal R )(\psi_t)\|_\infty
\le \| \mathcal R(\mu_t)-\mathcal R(\psi_t)\|_\infty + \|\mathcal R(\psi_t)-(\mathcal R_a\mathcal R )(\psi_t)\|_\infty.
\end{equation*}
The error has two terms. The first term is upper bounded as follows
\begin{equation*}
|\mathcal R(\mu_t)(\bar s)-\mathcal R(\psi_t)(\bar s)| 
\le\mathds 1_{\Upsilon}(\bar s) \int_{\mathcal S}t_{\mathfrak s}(\bar s|s)|\mu_t(s)-\psi_t(s)|ds\le E_t\int_{\mathcal S}t_{\mathfrak s}(\bar s|s)ds
\le M_{\mathfrak f} E_t.
\end{equation*}
Let us focus on the second term: 
for any arbitrary state $\bar s\in\mathcal A_j$, $j\in\mathbb N_n$, we have 
\begin{align*}
|\mathcal R(\psi_t)(\bar s)&-(\mathcal R_a\mathcal R )(\psi_t)(\bar s)|
= \left|\int_{\mathcal S}t_{\mathfrak s}(\bar s|u)\psi_t(u)du - \frac{1}{\mathcal L(\mathcal A_j)}\int_{\mathcal A_j}\int_{\mathcal S}t_{\mathfrak s}(v|u)\psi_t(u)du dv\right|\\
& \le \int_{\mathcal S}\psi_t(u)\left|t_{\mathfrak s}(\bar s|u)-\frac{1}{\mathcal L(\mathcal A_j)}\int_{\mathcal A_j}t_{\mathfrak s}(v|u)dv\right|du\\
& \le \int_{\mathcal S}\psi_t(u)\frac{1}{\mathcal L(\mathcal A_j)}\int_{\mathcal A_j}\left|t_{\mathfrak s}(\bar s|u)-t_{\mathfrak s}(v|u)\right|dv du
\le \int_{\mathcal S}\psi_t(u) \lambda_{\mathfrak f}\delta du \le \lambda_{\mathfrak f}\delta.
\end{align*}
The summation of the two upper bounds leads to that in \eqref{eq:part_error_rec}.
\end{proof}

\begin{proof}[Proof of Theorem \ref{thm:error_dynamic}]
The definition of $\psi_t^{\mathfrak h}$ implies that $\psi_{t+1}^{\mathfrak h} = \Pi_\Upsilon\mathcal R_\Upsilon(\psi_t^{\mathfrak h})$, 
with $\psi_0^{\mathfrak h} = \mu_0$. 
Then \eqref{eq:Main_inequality} is true for $t=0$ with $E_0^{\mathfrak h} = 0$. 
Assume that $\|\mu_t-\psi_t^{\mathfrak h}\|_\infty\le E_t^{\mathfrak h}$; 
then for any $\bar s\in\Upsilon$, 
\begin{align*}
|\mu_{t+1}(\bar s)-\psi_{t+1}^{\mathfrak h}(\bar s)| & = |\mathcal R_\Upsilon(\mu_t)(\bar s)
-\Pi_\Upsilon\mathcal R_\Upsilon(\psi_t^{\mathfrak h})(\bar s)|\\
& \le |\mathcal R_\Upsilon(\mu_t)(\bar s)-\Pi_\Upsilon\mathcal R_\Upsilon(\mu_t)(\bar s)|
+ |\Pi_\Upsilon\mathcal R_\Upsilon(\mu_t)(\bar s) - \Pi_\Upsilon\mathcal R_\Upsilon(\psi_t^{\mathfrak h})(\bar s)|.
\end{align*}
The first term can be upper bounded based on the linearity of the operator $\Pi_\Upsilon$ as 
\begin{align*}
|\mathcal R_\Upsilon(\mu_t)(\bar s) & - \Pi_\Upsilon\mathcal R_\Upsilon(\mu_t)(\bar s)|
= \left|\int_\Upsilon t_{\mathfrak s}(\bar s|s)\mu_t(s)ds-\int_\Upsilon\Pi_\Upsilon(t_{\mathfrak s}(\bar s|s))\mu_t(s)ds\right|\\
& \le \int_\Upsilon \mu_t(s)|t_{\mathfrak s}(\bar s|s)-\Pi_\Upsilon(t_{\mathfrak s}(\bar s|s))|ds
\le \mathcal E^{\mathfrak h}\int_\Upsilon \mu_t(s)ds
\le \mathcal E^{\mathfrak h}.
\end{align*}
On the other hand, the second term is upper bounded as follows: 
\begin{align*}
|\Pi_\Upsilon\mathcal R_\Upsilon(\mu_t)(\bar s)-\Pi_\Upsilon\mathcal R_\Upsilon(\psi_t^{\mathfrak h})(\bar s)|
& \le \int_\Upsilon |\Pi_\Upsilon(t_{\mathfrak s}(\bar s|s))||\mu_t(s)-\psi_t^{\mathfrak h}(s)|ds\\
& \le E_t^{\mathfrak h}\int_\Upsilon |\Pi_\Upsilon(t_{\mathfrak s}(\bar s|s))|ds
\le E_t^{\mathfrak h}M_{\mathfrak f}^{\mathfrak h}. 
\end{align*}
The addition of the two bounds leads to the statement. 
\end{proof}

\begin{proof}[Proof of Theorem \ref{thm:exp_conv}]
Based on the recursion over the density functions, we have that 
\begin{align*}
& \pi_{t+1}(\bar s)\le\sup\{\pi_{t}(s),\,s\in\mathcal S\}\int_{\mathcal S}t_{\mathfrak s}(\bar s|s)ds\le M_{\mathfrak f} \sup\{\pi_{t}(s),\,s\in\mathcal S\}\\
&\Rightarrow \pi_{t}(s)\le \sup\{\pi_0(s),\,s\in\mathcal S\}M_{\mathfrak f}^t, \quad \forall t\in\mathbb N.
\end{align*}
Then $\{\pi_t\}$ converges to zero uniformly exponentially, with a rate $M_{\mathfrak f}$. 
With focus on the safety problem we have $\mathbb P\{s(u)\in \mathcal A \text{ for all } u\in\mathbb N_t\}\le \mathbb P\{s(t)\in \mathcal A \}$, where
\begin{align*}
& \mathbb P\{s(t)\in \mathcal A \} = \int_{\mathcal A} \pi_{t}(s) ds \le \mathcal L(\mathcal A)\sup_{s\in\mathcal S}\pi_{0}(s)M_{\mathfrak f}^t.
\end{align*}
Note that if 
$
\lim\limits_{t\rightarrow\infty}\mathbb P\{s(t)\in \mathcal S\} = \lim\limits_{t\rightarrow\infty} 1 = 1,
$
then the state-space cannot be bounded under the assumptions. 
\end{proof}

\newpage
\section{List of Symbols}
\begin{longtable}{lp{10.5cm}}
$\mathbb N  = \{1,2,3,\ldots\}$& the set of natural numbers\\
$\mathbb N_m = \{1,2,\ldots,m\}$& the finite set of natural numbers\\
$\mathbb Z_m = \{0,1,2,\ldots,m\}$& the finite set of non-negative numbers\\
$\mathscr M_{\mathfrak s}$ &  discrete-time Markov process\\
$\mathcal S$ & state space of the Markov process\\
$\mathcal B(\mathcal S)$ & Borel $\sigma$-algebra on the space $\mathcal S$\\
$T_{\mathfrak s}$ & stochastic kernel\\
$\mathcal P$ & probability measure on $\mathcal S$ for one-step transition of the process\\
$t_{\mathfrak s}(\bar s|s)$& conditional density function of the process\\
$\pi_0:\mathcal S\rightarrow \mathbb R^{\ge 0}$& density function of the initial state\\
$\pi_t:\mathcal S\rightarrow \mathbb R^{\ge 0}$& density function of $s(t)$\\
$\mathbb P$& probability measure on the product space $\mathcal S^{t+1}$\\
$\mathbf s(t) = \left[s(0),s(1),\ldots,s(t)\right]$& bold typeset is employed for vectors\\
$\mathscr M_{\mathfrak p}$& Markov chain for forward computation of the density function\\
$a,b,\sigma,\alpha$& parameters of the example\\
$w(\cdot)$& process noise\\
$[\beta_0,\gamma_0]\subset\mathbb R$& support of the initial density function $\pi_0$ in the example\\
$\phi_{\sigma}(u)$& Gaussian density function with zero mean and standard deviation $\sigma$\\
$\varGamma\subset\mathcal S^2,\,\Lambda_0\subset\mathcal S$& sets used in the truncation procedure\\
$\epsilon$ & threshold on truncation part of $t_{\mathfrak s}$: $t_{\mathfrak s}(\bar s|s)\le \epsilon$ for all $(s,\bar s)\in\mathcal S^2\backslash\varGamma$\\
$\varepsilon_0$ & threshold on truncation part of $\pi_0$: $\pi_0(s)\le\varepsilon_0$ for all $s\in\mathcal S\backslash\Lambda_0$\\
$\varepsilon_t$ & threshold on truncation part of $\pi_t$: 
$\pi_t(s)\le \varepsilon_t$ for all $s\in\mathcal S\backslash\Lambda_t$\\
$\lambda_0$& Lipschitz constants of $\pi_0(s)$\\
$\lambda_{\mathfrak f}$& Lipschitz constants of $t_{\mathfrak s}(\bar s|s)$ with respect to $\bar s$\\ 
$M_{\mathfrak f}$ & upper bound for the quantities $\int_{\mathcal S}t_{\mathfrak s}(\bar s|s)ds,\, \bar s\in\mathcal S$\\
$\|\cdot\|_{\infty}$ & infinity norm\\
$\mathds 1_A(\cdot)$ & indicator function of set $A$\\
$\Lambda_t$ & support set of $\pi_t(\cdot)$ obtained via a recursive procedure\\
$\Upsilon = \cup_{t=0}^N\Lambda_t$ & truncated state space\\
$\varXi(s)$ & set-valued map $\varXi:\mathcal S\rightarrow 2^{\mathcal S}$\\
$\kappa(t,M_{\mathfrak f})$ & constant equal to $t$ for $M_{\mathfrak f} = 1$ and to $(1-M_{\mathfrak f}^t)/(1-M_{\mathfrak f})$ for $M_{\mathfrak f}\ne 1$\\
$\mu_t(\cdot)$ & approximate density function after state space truncation\\
$\Lambda_t = [\beta_t,\gamma_t]$ & support set of $\pi_t(\cdot)$ in the example\\
$\Upsilon = \cup_{i=1}^n \mathcal A_i$ & selected partition\\
$n$ & partition size\\
$\delta$ & partition diameter\\
$\mathcal A_{n+1} = \mathcal S\backslash\Upsilon$ & partition set associated with absorbing state of Markov chain\\
$P = [P_{ij}]$ & transition probability matrix of the Markov chain\\
$\delta_{(n+1)j}$ & Kronecker delta function, equal to one for $j = (n+1)$ and zero otherwise\\
$\mathcal L(\cdot)$ & Lebesgue measure of a set\\
$\mathbf{p_0} = [p_0(1),\ldots,p_0(n+1)]$ & pmf of a Markov chain at $t=0$\\
$\mathbf{p_t} = [p_t(1),\ldots,p_t(n+1)]$ & pmf of a Markov chain at time $t$\\
$\psi_t$ & approximation of the density function $\pi_t$ after abstraction\\
$E_t$ & abstraction error related to partitioning\\
$\mathbb B(\mathcal S)$ &  space of bounded and measurable functions on $\mathcal S$\\
$\mathbb B(\Upsilon)$ & space of bounded and measurable functions on $\Upsilon$\\
$\mathcal R_\Upsilon$ & linear operator on $\mathbb B(\Upsilon)$ defined as  $\mathcal R_\Upsilon (f)(\bar s) = \int_{\Upsilon}t_{\mathfrak s}(\bar s|s)f(s)ds,$ for all $\bar s\in\Upsilon$\\
$\Phi = \{\phi_1(s),\ldots,\phi_h(s)\}$ & set of basis functions $\phi_j\in\mathbb B(\Upsilon)$\\
$\Psi = span\,\Phi$ & function space generated by $\Phi$\\
$\Pi_{\Upsilon}:\mathbb B(\Upsilon)\rightarrow \Psi$ & projection operator\\
$\mathcal E^{\mathfrak h}$&  upper bound for $\left\| \Pi_{\Upsilon}(t_{\mathfrak s}(\cdot|s)) - t_{\mathfrak s}(\cdot|s)\right\|_\infty$\\
$M_{\mathfrak f}^{\mathfrak h}$&  upper bound for $\int_{\Upsilon}\left|\Pi_{\Upsilon}(t_{\mathfrak s}(\bar s|s))\right| ds$, for all $\bar s\in\Upsilon$\\
$\psi_t^{\mathfrak h}$ & approximation of $\pi_t(\cdot)$ computed via higher-order methods\\
$E_t^{\mathfrak h}$ & error of higher-order approximation\\
$d$ & dimension of the state space\\
$\Pi_{\mathcal D}$ & interpolation operator 
projecting any function $f:\mathcal D\rightarrow \mathbb R$
to a unique function of $\Psi$ such that $\Pi_{\mathcal D}(f) = \sum_{j=1}^{h}\alpha_j\phi_j$\\
$\mathbf f = [f(s_i)]_{i \in\mathbb N_h}$ & $h$-dim. column vector containing values of $f$ at the interpolation points\\
$\boldsymbol{\alpha} = [\alpha_j]_{j \in\mathbb N_h}$ & $h$-dim. column vector containing interpolation coefficients\\
$\mathcal Q = [\phi_j(s_i)]_{i,j}$ & $h$-dimensional interpolation matrix corresponding to $\Pi_{\mathcal D}$\\
$\{\phi_{ij},\,j\in\mathbb N_h\}$ & set of basis functions for partition set $\mathcal A_i$\\
$\{s_{ij},\,j\in\mathbb N_h\}$ & set of interpolation points in the partition set $\mathcal A_i$\\
$\{\alpha_{ij},\,j\in\mathbb N_h\}$ & set of interpolation coefficients in the partition set $\mathcal A_i$\\
$\mathcal Q_i = [\phi_{ij}(s_{iv})]_{v,j\in\mathbb N_h}$& matrix representation of interpolation inside partition set $\mathcal A_i$\\
$\mathbb I_h$& $h$-dimensional identity matrix\\
$f|_{\mathcal A_i}$& function $f:\Upsilon\rightarrow \mathbb{R}$ with domain restricted to set $\mathcal A_i\subset\Upsilon$\\
$\alpha_{ij}^{t}$ & interpolation coefficients for $\psi_t^{\mathfrak h}$ used in Algorithm~\ref{algo:app_val}\\
$\beta^{t}_{uv}$ & defined as $\beta^{t}_{uv}\doteq\mathcal R_\Upsilon(\psi^{\mathfrak h}_{t-1})(s_{uv})$ and used in Algorithm~\ref{algo:app_val}\\
$P_{ij}^{uv}$ & defined as $P_{ij}^{uv}\doteq \int_{\mathcal A_i}t_{\mathfrak s}(s_{uv}|s)\phi_{ij}(s)ds$ and used in Algorithm~\ref{algo:app_val}\\
$P = [P(i,j)]_{i,j}$ & matrix with entries $P(i,j) = \int_{\mathcal A_i}t_{\mathfrak s}(s_j|s)ds$ used for piecewise constant approximation of the density functions in Algorithm~\ref{algo:app_val_abs}\\
$\boldsymbol{\alpha_t} = [\alpha_t(i)]_i$& a row vector used in Algorithm~\ref{algo:app_val_abs} for values of  $\psi_t^{\mathfrak h}$\\
$\delta_i$& diameter of the partition set $\mathcal A_i$, $\delta_i = \sup\{\|s-s'\|,\,s,s'\in\mathcal A_i\}$\\
$[a_i,b_i]$& partition set $\mathcal A_i$ for one-dimensional systems\\
$\mathcal M_h$& an upper bound for the quantity $ \left|\partial^h t_{\mathfrak s}(\bar s|s)/\partial \bar s^h\right|$, for all $s,\bar s\in\Upsilon$\\
$[a_{i1},b_{i1}]\times[a_{i2},b_{i2}]$& partition set $\mathcal A_i$ for two-dimensional systems \\
$\mathcal M_2^k,\mathcal M_3^k$& upper bounds on partial derivatives of $t_{\mathfrak s}$ in 2-dim. systems\\
$\mathcal M_2^i,\mathcal M_3^{ij},\mathcal M_3$&  upper bounds on partial derivatives of $t_{\mathfrak s}$ in 3-dim. systems\\
$p_s^N(\mathcal A)$ & safety probability over the set $\mathcal A$ with time horizon $N$ and initial state $s$\\
$p_{\pi_0}^N(\mathcal A)$ & safety probability over the set $\mathcal A$ with the initial state admitting the density function $\pi_0$\\
$W_t:\mathcal S\rightarrow\mathbb R^{\ge 0}$ & sub-density functions for forward computation of the safety probability\\
$V_t:\mathcal S\rightarrow[0,1]$ & value functions for backward computation of the safety probability\\
$E_{\mathfrak f} = \kappa(N,M_{\mathfrak f})\lambda_{\mathfrak f}\delta\mathcal L(\mathcal A)$ & error of forward computation of the safety probability\\
$E_{\mathfrak b} = \kappa(N,M_{\mathfrak b})\lambda_{\mathfrak b}\delta \mathcal L(\mathcal A)$ & error of backward computation of the safety probability\\
$\lambda_{\mathfrak b}$ & Lipschitz constant of $t(\bar s|s)$ with respect to $s$\\
$M_{\mathfrak b}$ & upper bound for $\int_{\mathcal A}t_{\mathfrak s}(\bar s|s)d\bar s$, for all $s\in\mathcal A$\\
$\mathscr M_{\mathfrak b}$ & finite-state Markov chain obtained via the backward abstraction approach\\
$P(s_i,s_j)$ & transition probabilities of the Markov chain $\mathscr M_{\mathfrak b}$
\end{longtable}

\end{document}